\documentclass{article}
\usepackage[]{graphicx}
\usepackage[]{color}
\usepackage{amsthm}
\usepackage{amssymb,latexsym,amsthm,amsmath}
\usepackage{verbatim}
\usepackage{amsfonts}
\usepackage{latexsym} 
\usepackage{tikz} 
\usetikzlibrary{shapes, shapes.geometric, arrows,positioning}
\usepackage{amsfonts}
\usepackage{amssymb,latexsym,amsmath}
\usepackage{enumitem}

\renewcommand{\le}{\leqslant}
\renewcommand{\leq}{\leqslant}
\renewcommand{\geq}{\geqslant}
\renewcommand{\ge}{\geqslant}
\makeatletter
\def\maxwidth{ %
  \ifdim\Gin@nat@width>\linewidth
    \linewidth
  \else
    \Gin@nat@width
  \fi
}
\makeatother

\definecolor{fgcolor}{rgb}{0.345, 0.345, 0.345}

\usepackage{framed}
\makeatletter
 {\par\unskip\endMakeFramed%
 \at@end@of@kframe}
\makeatother

\definecolor{shadecolor}{rgb}{.97, .97, .97}
\definecolor{messagecolor}{rgb}{0, 0, 0}
\definecolor{warningcolor}{rgb}{1, 0, 1}
\definecolor{errorcolor}{rgb}{1, 0, 0}
\usepackage{alltt}
\usepackage[utf8]{inputenc}
\usepackage{amsmath,amssymb,indentfirst,epsfig,url}
\usepackage{makeidx} % Diagrama
\usepackage{graphicx}
\usepackage{subcaption} % Subfigure
\usepackage{multicol} % Criar 2 ou mais colunas em tabelas
\usepackage{multirow} % Criar 2 ou mais linhas em tabelas
\usepackage{hyperref} % Suporte para hipertexto, links para referências e figuras
\hypersetup{colorlinks=true, linkcolor=red, citecolor=blue, filecolor=black, urlcolor=green} % Configurações dos links 
\usepackage{verbatim} % Inclui caracteres externos de sem alteração gerado
\usepackage[affil-it]{authblk} % afiliação com affil
\usepackage{natbib} % Referências bibliográficas e afins
\bibpunct[; ]{(}{)}{,}{a}{,}{;} % Formatar as citações no texto e a lista de referências
\usepackage{amsmath}% pra colocar tipo de convergencia

\newtheorem{lemma}{Lemma}
\newtheorem{theorem}{Theorem}[section]
\providecommand{\keywords}[1]{\textbf{\textit{Keywords: }} #1}

\begin{document}
	\baselineskip = 5.7mm  % Espaço entre linhas;
	%\pagestyle{myheadings} % Estilo da página;
	
	%%Título
	\title{\textbf{Order book dynamics with liquidity fluctuations: \\
			limit theorems and large deviations}}
	
         \author[1]{Helder Rojas \thanks{hmolina@santander.com.br}}
          \affil[1]{\small{Santander Bank, S\~ao Paulo, Brazil}}
          
         \author[2]{Artem Logachov \thanks{omboldovskaya@mail.ru}}
          \affil[2]{\small{Sobolev Institute of Mathematics, Siberian Branch of the Russian Academy of Science, Russia}}

         \author[3]{Anatoly Yambartsev\thanks{yambar@ime.usp.br}}
       \affil[3]{\small{Institute of Mathematics and Statistics, USP, S\~ao Paulo, Brazil}}
	\date{} 

	\maketitle
	
%%--------------------------------------------------
	\begin{abstract}

We propose a class of  stochastic models for a dynamics of limit order book with different type of liquidities. Within this class of models we study the one where a spread decreases uniformly, belonging to the class of processes known as a population processes with uniform catastrophes. The law of large numbers (LLN), central limit theorem (CLT) and large deviations (LD) are proved for our model with uniform catastrophes. Our results allow us to satisfactorily explain the volatility and  local trends in the prices, relevant empirical characteristics that are observed in this type of markets. Furthermore, it shows us how these local trends and volatility are determined by the typical values of the bid-ask spread. In addition, we use our model to show how large deviations occur in the spread and prices, such as those observed in flash crashes.
\\
\\
\keywords{Limit order book, Liquidity fluctuations,  Markov chains, Limit theorems, Large Deviations, Flash crash.}
\end{abstract}

%%--------------------------------------------------
\section{Introduction}

The ``order book'' (OB) refers to an electronic list used to describe the evolution of bid and ask prices and sizes  in high-frequency electronic markets, such as NYSE-ARCA, LSE or NASDAQ. The evolution of the OB results from the interaction of buy and sell orders through a rather complex dynamic process. Order book dynamics has been extensively studied in the market microstructure and econophysics literature (\cite{biais1995empirical}, \cite{smith2003statistical}, \cite{bouchaud2009markets}), more recently, based on empirical characteristics presented in these studies, several models for the evolution of the OB have been proposed, see for instance \cite{cont2010stochastic}, \cite{avellaneda2011forecasting}, \cite{cont2013price}, \cite{cont2019stochastic}. These models, which are Markovian queueing systems, they generally implicitly assume uninterrupted high liquidity, i.e., they assume a abundant availability of limited orders in the OB. In this high liquidity context, the prices are relatively stable with small temporary fluctuations and the bid and ask sizes at the top of the OB provide valuable information on this short-term price fluctuations. Therefore,  these models are mainly focused on the direction of the next price movement and provide good results and a more or less clear understanding of the price dynamics in these conditions. On the other hand, in various markets the prices are not as stable, on the contrary, the prices show great changes, and even in some cases present local  down trends, mainly caused by liquidity wells in the OB. These events have caused controversy in the use of OB as the primary mechanism for trading,  in particular events such as the  May 6th, 2010 Flash Crash (see Figure \ref{FC}) have prompted market observers to question the stability of OB. Furthermore, the occurrence of Mini Flash Crash up and down, that is to say, rapid and significantly large directional movements in the price of assets, today are quite common, see \cite{golub2012high}. These events generate temporary liquidity crises resulting in larger spreads. Therefore, it is of both practical and theoretical interest to better understand how price dynamics depend on the structure and fundamental parameters of the OB.
\\
\begin{figure}[h!]
		\centering
	           \includegraphics[width=10.5cm,height=6.5cm]{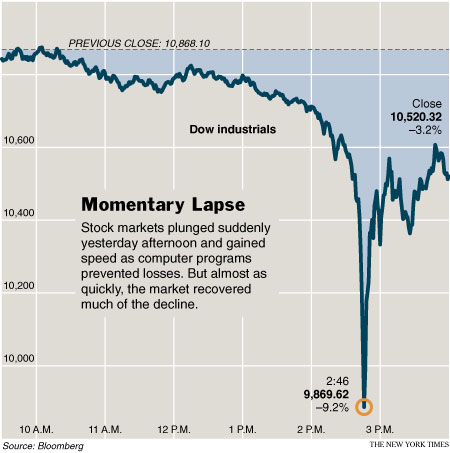}
		 \caption{A graph of the S$\&$P500 futures on the day of the Flash Crash of May 6th, 2010 at 2:45 pm.}
                      \label{FC}
	\end{figure}  
\\
In the present paper, we are interested in understanding how severe intermittencies in liquidity affect the order book dynamics. The contexts in which there are significant and intermittent decreases in the OB's ability to absorb market orders is what we call ``liquidity fluctuations''. Inspired in this context, we propose a simple model for price dynamics in an OB, our model explains in a simple way how large price fluctuations occur, fluctuations such as those observed in flash crashes. Our model successfully explains the local trends in the prices and the volatility around these trends expressed as function of parameters of the  micro-jumps of the prices. Furthermore, it shows us how these local trends and volatility are determined by the typical values of the bid-ask spread. From our price model, a model for the dynamics of the spread is implicitly derived, we use this model to analyze large deviations in the spread and its impact on prices, we present these large deviations in the form of ``optimal trajectories'' that give us relevant information about their occurrence. Finally, we present Monte Carlo simulations to corroborate that our model reproduces relevant empirical characteristics observed in our data as well as documented in the literature such as the famous bid-ask bounce, see \cite{roll1984simple}.

\subsection*{Motivation} 	
Our initial motivation arises from the price trends that are observed in various markets, those trends are local and eventually vary without any apparent pattern. Our interest was to understand the relationship between these observed long-term price trends with micro-jumps in short-term prices, see Figure \ref{short_long_prices}. The initial conjecture that motivated our work is the existence of a close relationship between the spread, price trend, volatility around these trends and the presence of liquidity fluctuations. Therefore, we needed a joint modeling of the spread and prices dynamics. In Figure \ref{short_long_prices}, we observe that the local trend is common to both the bid and ask prices. Therefore, this suggested to us that our model, in addition to presenting a long-term trend in prices, would have to incorporate the asymptotically stationary behavior of the spread. Usually large spread and price changes  are attributed to changes in liquidity (\cite{doyne2004really}). Our questions regarding the phenomenon grew and we set out to understand how large fluctuations in spread and prices occur, such as those observed in flash crashes, and how these rare events are related to liquidity fluctuations.
\\
\begin{figure}[h!]
		\centering
		\begin{subfigure}[b]{\linewidth}
			%\begin{subfigure}[h]{0.4\linewidth}
			\centering
	           	\includegraphics[width=0.8\linewidth]{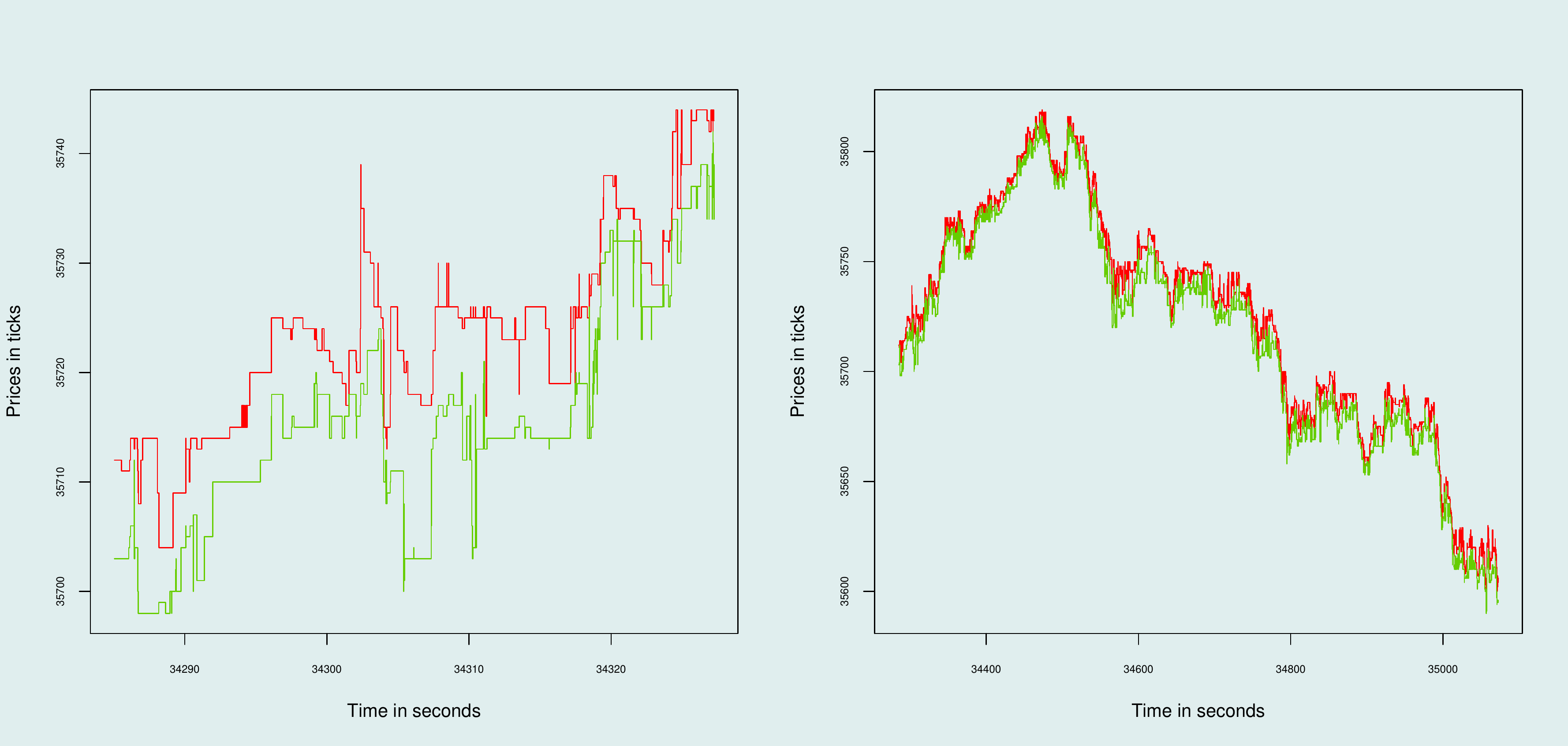}
		\end{subfigure} 
		\\
		\begin{subfigure}[b]{0.4\linewidth}
			\centering
			\includegraphics[width=\linewidth]{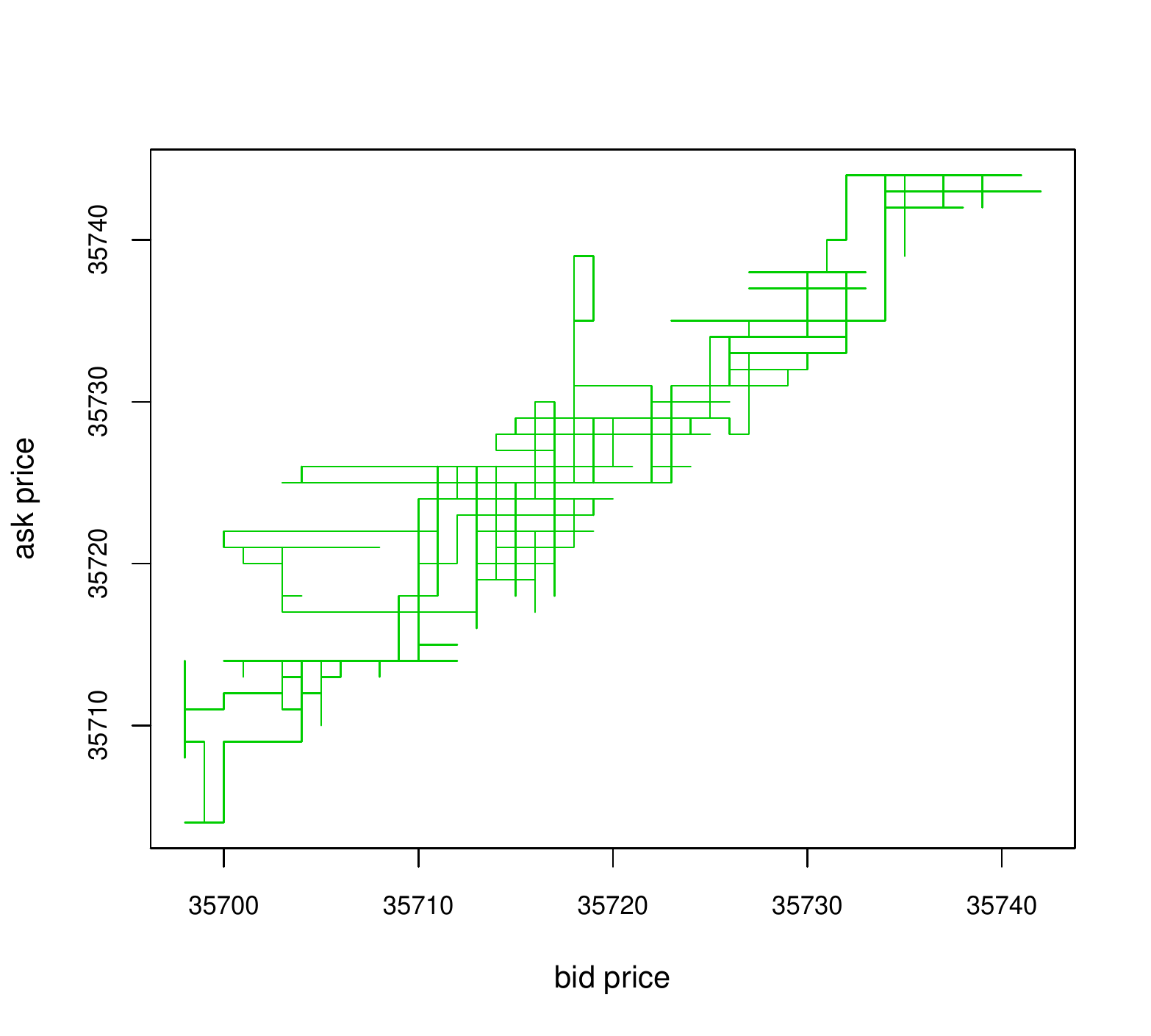} 
		\end{subfigure}
		\begin{subfigure}[b]{0.4\linewidth}
			\includegraphics[width=\linewidth]{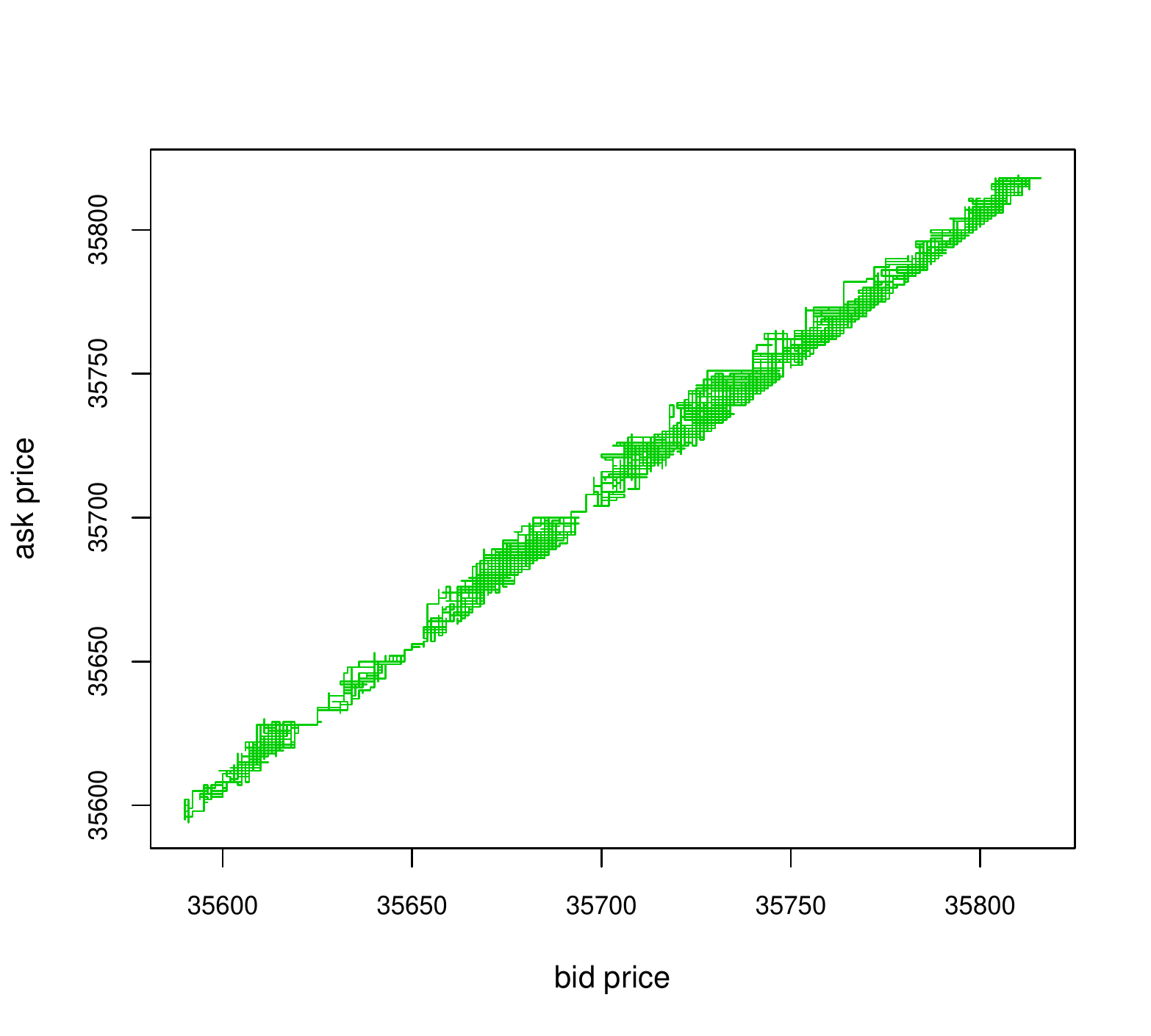}
		\end{subfigure}
		 \caption{Intraday evolution of the ask (red) and bid (green) price,
		 Apple Inc. (AAPL) stock, 04 March 2011. Left: Short-term, 1 minute. Right: Long-term, 15 minutes. }
         \label{short_long_prices}
	\end{figure}  
\\

%	\begin{figure}[h!]
%		\centering
%		\begin{subfigure}[b]{0.4\linewidth}
%			\includegraphics[width=\linewidth]{images/PricesShortRun.pdf} 
%		\end{subfigure}
%		\begin{subfigure}[b]{0.4\linewidth}
%			\includegraphics[width=\linewidth]{images/PricesLogRun.pdf}
%		\end{subfigure}
%    \caption{Intraday evolution of price process $X(t)=\big(P_{b}(t),P_{a}(t)\big)$,
%		 Apple Inc. (AAPL) stock, 04 March 2011. Left: Short-term, 1 minute. Right: Long-term, 15 minutes.}
%     \label{PricesSLR}
%	\end{figure}  

%		\centering
%		\begin{subfigure}[b]{0.35\linewidth}
%			\includegraphics[width=\linewidth]{images/Simul_shortprice_1.pdf} 
%		\end{subfigure}
%		\begin{subfigure}[b]{0.35\linewidth}
%			\includegraphics[width=\linewidth]{images/Simul_longprice_1.pdf}
%		\end{subfigure}

\subsection*{Outline} 
 Our paper is organized as follows. In Section \ref{OB_fluctuations} we describe some relevant empirical characteristics in a order book with liquidity fluctuations for which we consider three regimes based on the spread reversal process. In Section \ref{markov_model}  we present our general Markovian model for a order book. In Section \ref{regimen_principal} we present our results for a highly competitive regime. In Section \ref{other_regimens} we formulate the other two liquidity regimes that generalize the regime presented in the previous section.

\section{Markov model and regimes in liquidity fluctuations}\label{OB_fluctuations}

A very important empirical characteristic observed in markets with liquidity fluctuations is the low availability of orders in the OB, the queue sizes at the top of the OB are small most of the time, see e.g. Figure \ref{Q_bid_ask}. In this contexts, the queues sizes of the best bid and ask prices are no longer the determining factors in the dynamics of prices, for more details see \cite{doyne2004really}. If the liquidity intermittency is severe, even ``gaps" are formed in the OB (block of adjacent price levels that do not contain quotes). In these cases, the distribution of price changes is mainly determined by the distribution of the gap sizes in the OB. Taking these facts into account, if our interest is to explain the observed long-term price trends, we can focus only on micro-jumps in prices and disregard the size of the queues.
 \begin{figure}[h!]
		\centering
	           \includegraphics[scale=0.5]{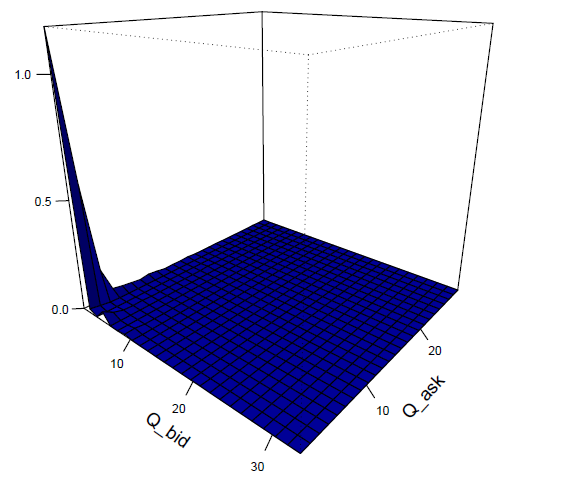}
		 \caption{Joint empirical distribution of bid and ask queue sizes at the top of the order book, Apple Inc. stock, 04 March 2011.}
                      \label{Q_bid_ask}
\end{figure}  

In these liquidity regimes, the spread exhibits a quite flexible dynamic behavior, reaching values much larger than those observed in high liquidity conditions, see e.g. Figure \ref{spread_flexibility}. Based on our empirical experiences, the OB slowly digests liquidity fluctuations and we can characterize that process in two stages. In a first stage, the spread begins to increase persistently. In the later stage, the spread is reduced, the reduction can be drastic or gradual. The closing type of the spread in the second stage depends on the intensity of the liquidity fluctuation.

 \begin{figure}[h!]
		\centering
	            \includegraphics[width=9.5cm,height=4.5cm]{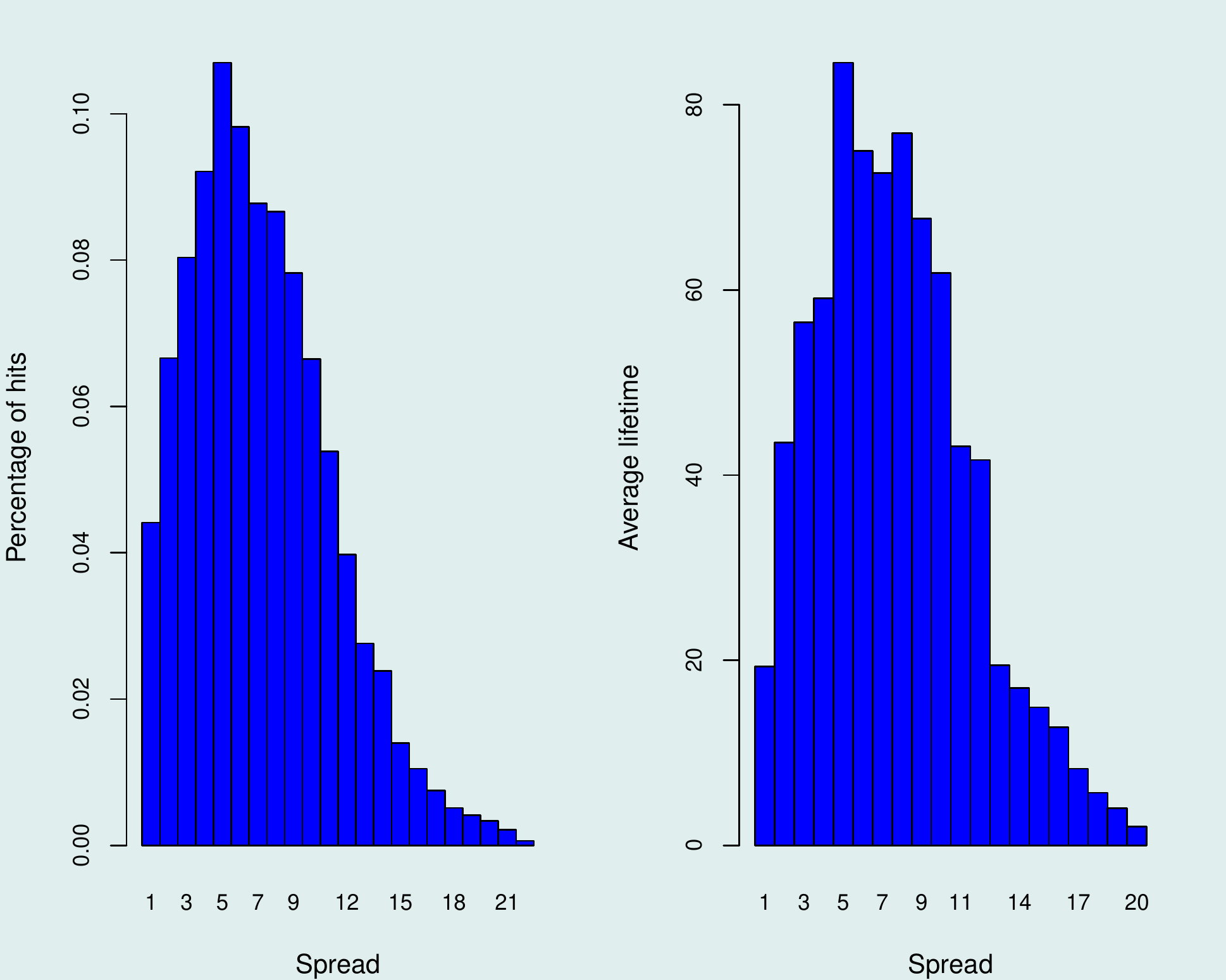} 
		 \caption{Empirical distribution of the bid-ask spread, Apple Inc. stock, 04 March 2011, corresponding to 15 minutes of observation. Left: Percentage of hit for each state. Right: Average lifetime of the spread, in seconds.}
                      \label{spread_flexibility}
	\end{figure}

Our empirical observations about the reversing process of  the bid-ask spread to its typical values, before and after liquidity shocks, have been theoretically corroborated through equilibrium models, see \cite{biais2009liquidity}. In this paper we consider different types of reversing process of the spread, i.e., we consider three low liquidity regimes: Highly competitive, Non-competitive and Low liquidity with gaps. These three regimes correspond to low liquidity regimes but differ in the closing type of the spread and in the gaps presence in the OB.
 
%%--------------------------------------------------

\subsection{The Markov model: a general view}\label{markov_model}

Based on our discussions in the previous sections, we propose a simplified representation for an OB with liquidity fluctuations. Let $P_{b}(t)$  the (best) bid price and $P_{a}(t)$ the (best) ask price, the state of the OB is described by a continuous-time process $X(t)=\big(P_{b}(t),P_{a}(t)\big)$ which takes values in the discrete state space $\mathbb{N}_{\tau}\times\mathbb{N}_{\tau}=\tau\mathbb{N}\times\tau\mathbb{N}$ (two-dimensional lattice), where $\tau$ is the ``tick size'' and as usual $\mathbb{N}$ is the set of positive integers $\mathbb{N} = \{ 1, 2, \dots \}$. For simplicity we assume $\mathbb{N}\times\mathbb{N}$ as the state space of $X(t)$ but we interpret each of its states as a multiple of $\tau$. The price process $X(t)$ presents piecewise constant sample paths whose transitions correspond to the order book events that cause price variations, see e.g.  Figure ~\ref{short_long_prices}. Our goal is to find asymptotic behaviour of the prices process $X(t)$ as the result of many micro-jumps.

%	\begin{figure}[h!]
%		\centering
%		\begin{subfigure}[b]{0.4\linewidth}
%			\includegraphics[width=\linewidth]{images/PricesShortRun.pdf} 
%		\end{subfigure}
%		\begin{subfigure}[b]{0.4\linewidth}
%			\includegraphics[width=\linewidth]{images/PricesLogRun.pdf}
%		\end{subfigure}
%    \caption{Intraday evolution of price process $X(t)=\big(P_{b}(t),P_{a}(t)\big)$,
%		 Apple Inc. (AAPL) stock, 04 March 2011. Left: Short-term, 1 minute. Right: Long-term, 15 minutes.}
%     \label{PricesSLR}
%	\end{figure}  

%%--------------------------------------------------
%\section{The Markov model}\label{markov_model}

Based on this simplified representation, consider a continuous-time Markov chain $X(t)=\big(P_{b}(t),P_{a}(t)\big)$ with state space $\mathbb{X}\subset \mathbb{N}\times\mathbb{N}$
$$
\mathbb{X} =\{ (b, a)\in \mathbb{N}\times\mathbb{N}:\ \mbox{ such that } b < a \}.
$$
Here $P_{b}(t)$ represents the bid price, $P_{a}(t)$ represents the ask price, and $S(t)=P_{a}(t)-P_{b}(t)$ is the bid-ask spread. The transitions of the chain $X(t)$ defined by the transition rates: let $(b,a)$ be a state of Markov chain, then
\begin{equation}\label{model}
\begin{aligned}
& (b,a) \to  (b,a+\Delta)\hspace{0.5cm} \textrm{with rate} \hspace{0.5cm} \alpha_{+}(\Delta), \\
& (b,a) \to  (b,a-\Delta)\hspace{0.5cm} \textrm{with rate} \hspace{0.5cm} \alpha_{-}(\Delta), \hspace{0.4cm} \textrm{where } 0 <\Delta < b-a,\\
& (b,a) \to  (b-\Delta,a)\hspace{0.5cm} \textrm{with rate} \hspace{0.5cm} \beta_{-}(\Delta), \\
& (b,a) \to  (b+\Delta,a)\hspace{0.5cm} \textrm{with rate} \hspace{0.5cm} \beta_{+}(\Delta), \hspace{0.4cm} \textrm{where } 0 <\Delta < b-a,
\end{aligned}
\end{equation}
in all cases, increment $\Delta$ is a positive integer number. The function $\alpha_{+}(\cdot)$ (resp. $\beta_{-}(\cdot)$) is the rate at which increases (resp. decreases) in the ask (resp. bid) price occur as a result of the execution of market buy (resp. sell) orders or cancellations of limited sell (resp. buy) orders, as well as, that the functions $\alpha_{-}(\cdot)$ (resp. $\beta_{+}(\cdot)$) is the rate at which the decreases (resp. increases) in the ask (resp. bid) price occur as a result of a limited sell (resp. buy) order placed within the spread.

We study the asymptotic behavior of $X(t)$ as $t$ goes to infinity. In order to do this, it is convenient consider an equivalent process  $Y(t)=\big(P_{b}(t),S(t)\big)$ with state space $\mathbb{Y}=\mathbb{N} \times \mathbb{N}$.
Although $X(t)$ and $Y(t)$ contain the same information, the second representation gives us greater control in the asymptotic analysis. The transitions of the chain $Y(t)$ defined by the transition rates: let $(b,s)$ be a state of Markov chain, then
\begin{equation}\label{model1}
\begin{aligned}
& (b,s) \to  (b,s+\Delta) &\textrm{with rate} \hspace{0.5cm}& \alpha_{+}(\Delta), \\
& (b,s) \to  (b,s-\Delta) &\textrm{with rate} \hspace{0.5cm}& \alpha_{-}(\Delta), \\
& (b,s) \to  (b-\Delta,s+\Delta) &\textrm{with rate} \hspace{0.5cm}& \beta_{-}(\Delta), \\
& (b,s) \to  (b+\Delta,s-\Delta) &\textrm{with rate} \hspace{0.5cm}& \beta_{+}(\Delta).
\end{aligned}
\end{equation}
Since the transition rates of $Y(t)$ depend only on the second coordinate, the spread, thus we see that $S(t)$ alone is the continuous-time Markov process and has the following transition rates: suppose that at some moment the spread is $k \in \mathbb{N}$, then 
\begin{equation}\label{modelS}
\begin{aligned}
k & \to k+\Delta & \mbox{ with rate } & \ \ \gamma_+(\Delta) := \alpha_+(\Delta) + \beta_-(\Delta),\\ 
k & \to k-\Delta & \mbox{ with rate } & \ \ \gamma_-(\Delta) := \alpha_-(\Delta) + \beta_+(\Delta).
\end{aligned}
\end{equation}

Based on the model (\ref{model}) and their alternative representation (\ref{model1}) and (\ref{modelS}), we consider the three low liquidity regimes: highly competitive, non-competitive and low liquidity with gaps. Any regime is defined by the how the rates depend on increment $\Delta$ which are usually determined by the intensity of liquidity fluctuations. In this paper we focus on the first regime, highly competitive regime, the other two regimes are succinctly formulated. The results presented here can be generalized for the other two regimes, but we believe that qualitatively there will be no significant difference in the results.

%%--------------------------------------------------
\section{The Markov model: closing the spread uniformly (highly competitive regime)}\label{regimen_principal}

The highly competitive regime (HC regime) characterized by very small opening steps of the spread and a very rapid decrease in it. This regime is consistent with a rapid reversing process of the spread and the absence of gaps in the OB. The very rapid decreasing of spread is caused by the competitive behavior of impatient agents that place quotes within the spread, prioritizing the negotiations of their placed limited orders. Thus, in the considered model (\ref{model}) we define the rates in such a way that the spread can increase only by one unit and for a given length of the spread, say $k$, the next length of the spread is chosen uniformly from the set $I_{k}=\{1, \dots, k-1\}$.

In order to define the rates for highly competitive regime, we abuse our notations and we fix the parameters $\alpha_{+}, \alpha_{-}, \beta_{+}, \beta_{-}$, which are strictly positive real numbers. Further, the notations $\alpha_\pm, \beta_\pm$ are used only as a parameters of the model and not as the functions. Define the transition rates for the Markov chain $X(t)$ in the following way: suppose that at some moment the chain is at some state $(b,a)\in \mathbb{X}$, then 
\begin{equation}\label{HLmodel}
\begin{split}
\alpha_+(\Delta) = \left\{ \begin{array}{rl} \alpha_{+}, & \mbox{ if }\Delta =1;\\ 0, & \mbox{ otherwise}; \end{array}\right.
\ \ 
\alpha_-(\Delta) = \left\{ \begin{array}{rl} \frac{\alpha_{-}}{a-b-1}, & \mbox{ if } b-a >1\mbox{ for any }\Delta \in I_{a-b}; \\ 0, & \mbox{ otherwise}; \end{array}\right.
\\ 
\beta_-(\Delta) = \left\{ \begin{array}{rl} \beta_{-}, & \mbox{ if }\Delta =1;\\ 0, & \mbox{ otherwise}; \end{array}\right.
\ \
\beta_+(\Delta) = \left\{ \begin{array}{rl} \frac{\beta_{+}}{a-b-1}, & \mbox{ if } b-a >1\mbox{ for any }\Delta \in I_{a-b}; \\ 0, & \mbox{ otherwise}. \end{array}\right.
\end{split}
\end{equation}
For illustration see Figure~\ref{HLCmodelFig} in the case when $a-b = 3$.

	\begin{figure}[h!]
		\centering
	           \includegraphics[scale=0.6]{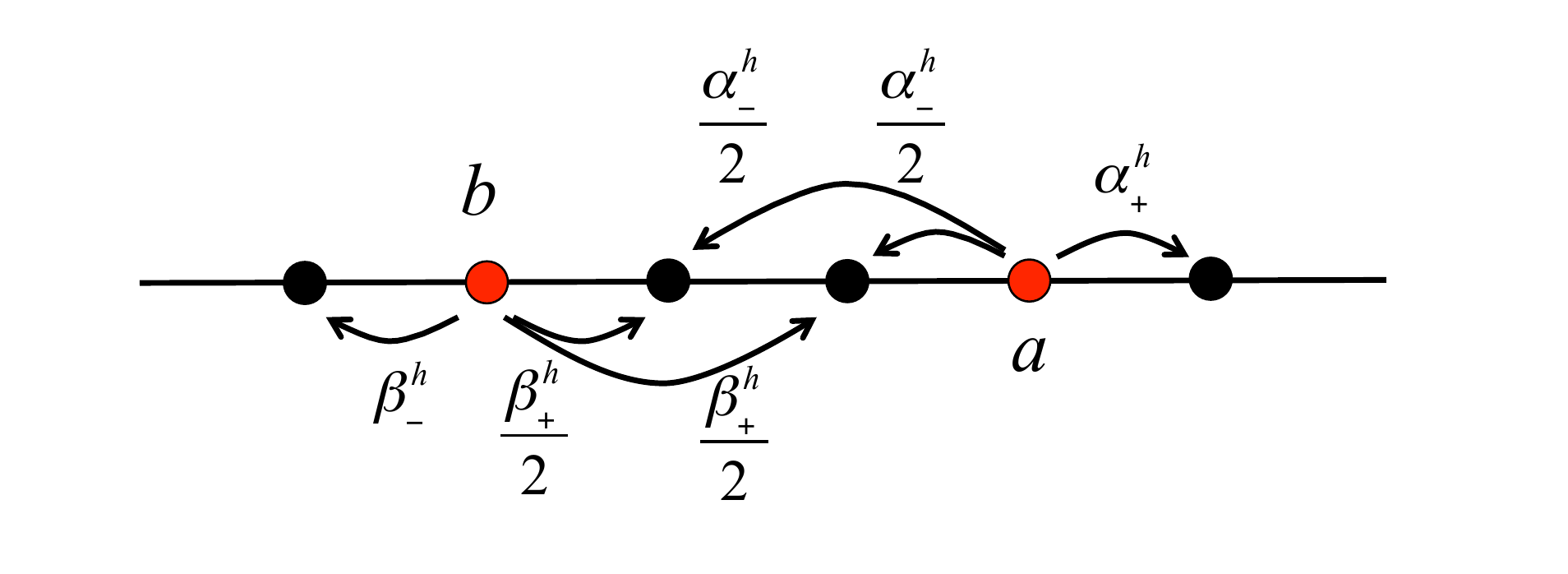} 
		 \caption{The rates for Highly Competitive model. Illustrative example for the case when $a-b = 3$.}
                      \label{HLCmodelFig}
	\end{figure}  

In this regime the transition rates of $S(t)$, see (\ref{modelS}), are the following: suppose that at some moment the spread is $k \in \mathbb{N}$, and let $\gamma_+:=\beta_{-}+\alpha_{+}$ and $\gamma_-:=\beta_+ + \alpha_{-}$, then 

\begin{equation}\label{HLmodelS}
\begin{aligned}
k & \to k+1 & \mbox{ with rate } & \ \ \gamma_+,\\ 
k & \to k-\Delta & \mbox{ with rate } & \ \ \frac{\gamma_-}{k-1}  \textrm{ for }  \Delta \in I_{k}.
\end{aligned}
\end{equation}
Note that $S(t)$ is irreducible Markov chain in this regime. 

\subsection{Ergodicity and invariant measure for $S(t)$} 

We start with $S(t)$. The next theorem provide ergodicity and, it is one of the rare cases when we can find the invariant measure for the process.

\begin{theorem}
In highly competitive regime model, for any positive values of parameters $\alpha_{+}, \alpha_{-}, \beta_{+}, \beta_{-}$ the spread $S(t)$ is a positive recurrent Markov process with invariant measure $\mu = \left(\mu(k), k\in \mathbb{N}\right)$ is the following: let $\gamma := \gamma_+ + \gamma_-$
\begin{equation}\label{inv}
\mu(k) = \frac{k!(\gamma_+)^{k-1}}{ \prod_{i=1}^{k-1} (\gamma_- + i \gamma) } \left( 1 + \sum_{k\geq 2} \frac{k!(\gamma_+)^{k-1}}{ \prod_{i=1}^{k-1} (\gamma_- + i \gamma) }\right)^{-1}.
\end{equation}
\end{theorem}

\begin{proof} 
In order to prove the positive recurrence we provide the Lyapunov function following the criteria for continuous-time Markov chains from \cite{menshikov2014explosion}, Theorem 1.4. The positive recurrence is equivalent to the existence of a non-negative function $f$ (Lyapunov function) on the set of states, some small positive $\varepsilon$ and a finte set of states $F$, such that applying the generator of the process $\Gamma$ to $f$ we obtain $\Gamma f(x) < - \varepsilon,$ for all $x\notin F$.

Recall, that the generator $\Gamma$ is the matrix $\Gamma=(\Gamma_{xy})$ where $\Gamma_{xy}, x\ne y$ are transition rates and $\Gamma_{xx} = - \sum_{y\ne x} \Gamma_{xy}$. Let $f(x) = x$. Applying the criteria:
$$
\Gamma f(k) = \gamma_+(k+1) + \sum_{x\in I_k} \frac{\gamma_+}{k-1} x - (\gamma_++\gamma_-)k = \gamma_+ - \frac{\gamma_- k}{2} < - 1,
$$
for all $k: k > 2(\gamma_++1)/\gamma_-$. The last inequality provides the finite set $F=\{ k \le 2(\gamma_++1)/\gamma_-\}$.

The formula for invariant measure can be checked directly by the global balance equations. Indeed, starting with system of the global balance equation, we obtain following.
$$
\left\{ 
\begin{aligned}
\gamma_+ \mu (1) & =  \gamma_-\mu(2) + \frac{\gamma_-}{2}\mu(3) + \frac{\gamma_-}{3}\mu(4) + \dots \\
(\gamma_+ + \gamma_-) \mu (2) & = \gamma_+ \mu(1) + \frac{\gamma_-}{2}\mu(3) + \frac{\gamma_-}{3}\mu(4) + \dots \\
(\gamma_+ + \gamma_-) \mu (3) & = \gamma_+ \mu(2) + \frac{\gamma_-}{3}\mu(4) + \frac{\gamma_-}{4}\mu(5) + \dots \\
\dots  \\
(\gamma_+ + \gamma_-) \mu (n) & = \gamma_+ \mu(n-1) + \frac{\gamma_-}{n}\mu(n+1) %+ \frac{\gamma_-}{n+1}\mu(n+2) 
+ \dots \\
\dots  \\
\end{aligned}
\right.
\Leftrightarrow
\left\{ 
\begin{aligned}
2 \gamma_+\mu (1) & = (\gamma_++2\gamma_-)\mu(2) \\
(2\gamma_+ + \gamma_-)\mu (2) & = \gamma_+\mu(1) + (\gamma_++\gamma_-+\frac{\gamma_-}{2})\mu(3) \\
(2\gamma_+ + \gamma_-)\mu (3) & = \gamma_+\mu(2) + (\gamma_++\gamma_-+\frac{\gamma_-}{3})\mu(4) \\
\dots \\
(2\gamma_+ + \gamma_-)\mu (n) & = \gamma_+\mu(n-1) + (\gamma_++\gamma_-+\frac{\gamma_-}{n})\mu(n+1) \\
\dots \\
\end{aligned}
\right.
$$
Dividing left side and right side on $\gamma_++\gamma_-$ and using the notation $p:=\frac{\gamma_+}{\gamma_++\gamma_-}, q:=1-p$ we rewrite the last system as
$$
\left\{ 
\begin{aligned}
\mu (2) & = \frac{2p}{1+q} \mu(1) \\
\mu (3) & = \frac{1+p}{1+\frac{q}{2}} \mu(2) - \frac{p}{1+\frac{q}{2}} \mu(1) \\
\dots \\
\mu (n) & = \frac{1+p}{1+\frac{q}{n-1}} \mu(n-1) - \frac{p}{1+\frac{q}{n-1}} \mu(n-2) \\
\dots \\
\end{aligned}
\right.
\Leftrightarrow
\left\{ 
\begin{aligned}
\mu (2) & = \frac{2p}{1+q} \mu(1) \\
\mu (3) & = \frac{3p^2}{(1+q)(1+\frac{q}{2})} \mu(1) \\
\dots \\
\mu (n) & = \frac{np^{n-1}}{ \prod_{i=1}^{n-1} (1+\frac{q}{i})} \mu(1) \\
\dots \\
\end{aligned}
\right.
$$
After the normalizing of the relation
\begin{equation}\label{inv1}
\mu (n) = \frac{np^{n-1}}{ \prod_{i=1}^{n-1} (1+\frac{q}{i})} \mu(1)
\end{equation}
and remembering that $\mu$ is probability measure, $\sum \mu(k) =1$, we return the notations $\gamma_\pm$ and obtain the formula (\ref{inv}).
\end{proof}

The process $S(t)$ belongs to the class of processes known as the population processes with uniform catastrophes. 
%In \cite{logachov2019large} the following limiting property of the maximum of the spread was proved: 
%$$
%\mathbb{P} \left( \lim_{T\to\infty} \sup_{t\in [0,1]} \frac{S(tT)}{T} > \varepsilon \right) = 0.
%$$
An extension for the processes with \textit{almost uniform} catastrophes (see Section~\ref{almostUC} for definition) was considered in \cite{logachov2018local}, where the following result was proved for the maximum of the process: for any fixed $b\in (0,1)$
$$
\mathbb{P} \left( \lim_{T\to\infty} \sup_{t\in [0,1]} \frac{S(tT)}{T^b} > \varepsilon \right) = 0.
$$

These results of the stationary asymptotic behavior of spread process $S(t)$ are consistent with the empirical observations presented in Figures \ref{short_long_prices} and \ref{spread_flexibility}.

%%--------------------------------------------------
%\noindent
\subsection{Local drift (LLN for the prices)}
The next theorem is about the Law of Large Numbers (LLN) for prices. This theorem will shed light on local trends (local drift) of the prices. 

\begin{theorem}\label{LLNh}
With probability one
\begin{equation}
\frac{P_{b}(t)}{t} \to D \hspace{0.5cm} \textrm{a.s.} \hspace{0.5cm} t \to \infty
\end{equation}
where
$$
 D = - \beta_- - \mu(1)\gamma\Bigl( \frac{\beta_-\gamma_-}{\gamma_+} + \frac{\beta_+}{2} \Bigr) +  \frac{\beta_+\gamma}{2} \sum_{k\ge 1} k\mu(k).%\mathbb{E}_{\mu}S(t).
$$
\end{theorem}

\begin{proof}
Probably, the most simple way to prove the LLN is due the ergodic theorem for discrete-time Markov chains. Indeed, let $s_n$ be the imbedded (of $S(\cdot)$ discrete-time Markov chain on $\mathbb{N}$ with following transition probabilities: 
$$
p(k,l):=\mathbb{P}(s_{n+1} = l \mid s_n = k) = \left\{ \begin{aligned}
1, & \mbox{ if } l = 2 \mbox{ and } k = 1, \\
\frac{\gamma_+}{\gamma_++\gamma_-}, & \mbox{ if }  l= k+1, \mbox{ when } k > 1, \\
\frac{\gamma_-}{\gamma_++\gamma_-} \frac{1}{k-1}, & \mbox{ if } l \in I_k \mbox{ and } k > 1. 
\end{aligned} \right.
$$
%From the definition of the original model (): 
%$$
%\gamma_+ = \alpha_+ + \beta_- \mbox{ and } \gamma_- =  \alpha_- + \beta_+.
%$$
Let $\pi=\big(\pi(s), s\in \mathbb{N})$ be the stationary measure for the chain $s_n$. 
%Both stationary measures are related as follows: $\mu(s)=\frac{\gamma_+}{\gamma_++\gamma_-} \pi(s)$, for all $s\in \mathbb{N}^{*}$, see \cite{norris1998markov}. 
It is easy to check that the relation (\ref{inv}) for the stationary measure $\mu$ will be transformed into the relation 
\begin{equation}\label{inv1}
\pi (n) = \frac{np^{n-2}}{ \prod_{i=1}^{n-1} (1+\frac{q}{i})} \pi(1), \ \ n\ge 2,
\end{equation}
for the stationary measure $\pi$. It can be checked directly from the global balance equations for discrete time Markov chain $s_n$.

Denote $p_n$ the discrete-time embedding chain corresponding to the bid-price continuous-time dynamics of $P_b(t)$. The dynamics of $p_n$ can be represented as a function of the the dynamics of the spread $s_n$, by the following formula. 
\begin{equation}\label{eq2}
p_n = \sum_{i=0}^n F(s_{n-1}, s_n, U_n),
\end{equation}
where $U_n$ is the sequence of i.i.d. uniformly distributed random variables, and where the function $F$ is
\begin{equation}\label{f}
F(s_{n-1}, s_n, U_n) = \left\{ \begin{aligned} 
-1, & \mbox{ if } s_{n} = s_{n-1} + 1 \mbox{ and } U_n < \frac{\beta_-}{\gamma_+}; \\
s_{n-1} - s_{n},  & \mbox{ if } s_{n} < s_{n-1} \mbox{ and } U_n < \frac{\beta_+}{\gamma_-}; \\
0, & \mbox{ otherwise}.
\end{aligned}
\right.
\end{equation}
Note that $\hat{s}_n = (s_{n-1}, s_n, U_n)$ is a Markov chain, and let $\hat\pi$ be their invariant measure. Observe that the discrete part of invariant measure $\hat{\pi}$ for the process $(s_{n-1},s_n)$ is the product $\hat{\pi}(x,y)=\pi(x)p(x,y)$
 By ergodic theorem we obtain LLN for embedding chain $p_n$.
\begin{lemma}\label{LLN}
%With probability one
\begin{equation}\label{eq1}
\frac{p_n}{n} \to v  \hspace{0.5cm} \mbox{ a.s.} \hspace{0.5cm} n \to \infty
\end{equation}
where
$$
v = - \frac{\beta_-}{\gamma_++\gamma_-} - \frac{\pi(1)}{\gamma_++\gamma_-} \Bigl( \frac{\beta_-\gamma_-}{\gamma_+} + \frac{\beta_+}{2} \Bigr) + \frac{\beta_+}{\gamma_++\gamma_-} \frac{1}{2}  \sum_{s=1}^\infty s\pi(s).
$$
\end{lemma}

\begin{proof}
The ergodic theorem states the convergence (\ref{eq1}). Thus we need only to find the $v$, which is the expectation over the invariant measure $\hat\pi$ of the increments $F(\hat{s}_n)$:
$$
\begin{aligned}
\mathbb{E}_{\hat\pi} (F(\hat{s}_n)) = &  -\pi(1) \frac{\beta_-}{\gamma_+} - \sum_{s=2}^\infty \pi(s) \frac{\gamma_+}{\gamma_++\gamma_-} \frac{\beta_-}{\gamma_+} + \sum_{s_2=2}^\infty \sum_{s_1=1}^{s_2-1} (s_2-s_1) \pi(s_2) \frac{\gamma_-}{\gamma_++\gamma_-}\frac{1}{s_2-1} \frac{\beta_+}{\gamma_-} \\
%= & - \frac{\beta_-}{\gamma_++\gamma_-} - \frac{\gamma_-}{\gamma_++\gamma_-} \frac{\beta_-}{\gamma_+}\mu(1) + \frac{\gamma_-}{\gamma_++\gamma_-} \frac{\beta_+}{\gamma_-} \sum_{s_2=2}^\infty\mu(s_2) \frac{1}{s_2-1}  \cdot \frac{s_2(s_2-1)}{2} \\
= & - \frac{\beta_-}{\gamma_++\gamma_-} - \frac{\pi(1)}{\gamma_++\gamma_-} \Bigl( \frac{\beta_-\gamma_-}{\gamma_+} + \frac{\beta_+}{2} \Bigr) + \frac{\beta_+}{\gamma_++\gamma_-} \frac{1}{2}  \sum_{s=1}^\infty s\pi(s),
\end{aligned}
$$
which finishes the proof of the lemma. $\Box$
\end{proof}

To finish the proof of Theorem~\ref{LLNh} we observe
$$
\begin{aligned}
\lim_{t\to\infty}\frac{P_b(t)}{t} &= \lim_{t\to\infty} \frac{p_{N_t}}{N_t} \frac{N_t}{t} = v (\gamma_++\gamma_-) \\
& = - \beta_- - \pi(1)\Bigl( \frac{\beta_-\gamma_-}{\gamma_+} + \frac{\beta_+}{2} \Bigr) +  \frac{\beta_+}{2}  \sum_{s=1}^\infty s\pi(s) =: D, %\\
%& = - \beta_- - \pi(1)\gamma\Bigl( \frac{\beta_-\gamma_-}{\gamma_+} + \frac{\beta_+}{2} \Bigr) +  \frac{\beta_+}{2}\gamma \sum_{s=1}^\infty s\pi(s) =: D,
\end{aligned}
$$
where $N_t$ is the Poisson process with rate $\gamma_++\gamma_-$.
\end{proof}

This result confirms our conjecture about the influence of the spread on the local trend of the prices. From a practical point of view, knowing the jump rate, of the bid and ask prices, we can calculate in a simple way the trend of the prices.
%%--------------------------------------------------
%\noindent
\subsection{Price volatility (CLT for the prices)}

In this section we are interested in studying a link between the price volatility  and the prices jump rates. In particular, we prove a Central Limit Theorem(CLT) for the price process. We express the volatility of price changes, around local drift, in terms of the jumps rate of ask and bid prices, i.e., for the process represented by (\ref{eq2}) the central limit theorem holds.

Again we first prove CLT for the imbedded discrete-time dynamics $p_n$ of the price, and then we establish CLT for continuous-time chain $P_b(t)$. 

\subsubsection{CLT for embedding $p_n$}

In order to prove it one way, for example, to prove that the chain $\hat{s}_n$ is geometrically ergodic, i.e., the rate of convergence to the invariant measure is geometric:
\begin{equation}\label{GeomErg}
\| P^n(x, \cdot) -\hat \pi(\cdot) \| \le M(x) q^n, \mbox{ for some } q<1,
\end{equation}
where $\|\cdot\|$ stands for total variation norm. And then apply the result for geometrically ergodic chains. Formally, we will require that the chain  $\hat{s}_n$ should be Harris ergodic Markov chain, that it is true for the chain $\hat{s}_n$. Here we refer to \cite{jones2004markov}.

\begin{theorem}\label{CLTGeomErg}
(See Theorem 9, \cite{jones2004markov}) Let $X$ be a Harris ergodic Markov chain on $\mathsf{X}$ with invariant distribution $\pi$ and let $f: \mathsf{X} \to \mathbb{R}$ is a Borel function. Assume that $X$ is geometrically ergodic and $E_\pi |f(x)|^{2+\delta} < \infty$ for some $\delta >0$. Then for any initial distribution, as $n\to\infty$ 
$$
\sqrt{n} (\bar{f}_n - E_\pi f) \to N(0, \sigma_f^2)
$$
in distribution.
\end{theorem}

Let us prove first, that the chain $\hat{s}_n$ is geometrically ergodic. Indeed, there are general results on the so-called drift conditions for the chain to be geometrically ergodic, see \cite{meyn2012markov}. But for the countable Markov chain we can apply the criteria of \cite{popov1977geometric}, see Theorem 2:

\begin{enumerate}
\item[] \textit{ A countable Markov chain is geometrically ergodic iff there exists a finite set $B\subset \mathsf{X}$ and function $g(x), x\in  \mathsf{X}$ such that $E_x e^{g(X_1) - g(X_0)} \le q < 1,$ when $x\notin B$ and $E_x e^{g(X_1) - g(X_0)} < \infty$, if $x\in B$.}
\end{enumerate}

Using this criteria we will find the corresponding function $g$ in the following form $V(\cdot) \equiv e^{g(\cdot)}$. 
To check the conditions we should provide the function $V(\cdot)$ for the chain $\hat{s}_n$: let
$V(i,j,u) = i^2 + j^2$, where the constant $c$ will be chosen later. Indeed, let for simplicity $p_+ = \frac{\gamma_+}{\gamma_++\gamma_-} $ and $p_- = \frac{\gamma_-}{\gamma_++\gamma_-}$, then
%$$
%\begin{aligned}
\begin{eqnarray}\label{e1} 
\mathbb{E} \left. \left( \frac{ V(s_{n}, s_{n+1}, U_{n+1}) }{ V(s_{n-1}, s_{n}, U_{n}) } \ \right|\ s_{n-1} = y, s_{n} = x \right) 
%\nonumber \\
& = &\sum_{k=1}^{x-1} \left( \frac{ (x-k)^2 + x^2 }{x^2 + y^2} \right) \frac{p_-}{x-1} + \left( \frac{ (x +1)^2 + x^2}{y^2 + x^2}  \right) p_+ \nonumber \\ 
& = &\frac{x(2x-1)}{6(x^2+y^2)} p_- + \frac{x^2}{x^2+ y^2}p_- +  \frac{ (x +1)^2 + x^2}{y^2 + x^2} p_+ 
%\bigl( \sum_{k=1}^{x-1} k^2 \bigl) \frac{p_-}{x-1} - x^2 p_- + (2x +1)p_+ 
%& < \bigl( \sum_{k=0}^{x-1} k^2 \bigl) \frac{p_-}{x} - p_-(2x-1)^2 + p_+((2x+1)^2 - (2x-1)^2) \\
%= - \frac{4x^2 + x}{6} p_- + (2x +1)p_+ \le - \frac{2}{3} p_- x^2,
\end{eqnarray}
%\end{aligned}
%$$
Consider two cases. First, we suppose that $x < y$. Then
$$
(\ref{e1})  \le \frac{x(2x-1)}{6(x^2+(x+1)^2)} p_- + \frac{x^2}{x^2+ (x+1)^2}p_- +  \frac{ (x +1)^2 + x^2}{x^2+ (x+1)^2} p_+ 
< \frac{2}{3} p_- + p_+  \le q < 1.
$$
Thus, the condition holds for all $(x,y)$ such that $x<y$. In the second case, $x>y$, we have $y=x-1$: 
\begin{eqnarray*}%\label{l1}
(\ref{e1})  &=& 
\frac{x(2x-1)}{6(x^2+(x-1)^2)} p_- + \frac{x^2}{x^2+ (x-1)^2}p_- +  \frac{ (x +1)^2 + x^2}{x^2+ (x-1)^2} p_+ \nonumber \\
&=& \frac{(8x^2 -x)p_- + 24x p_+}{6(x^2+(x-1)^2)} + p_+.
\end{eqnarray*}
There is no $q<1$ such that $(\ref{e1}) \le q$ for all $x$. But it is easy to see that there exists $q> \frac{2}{3} p_- + p_+$ and $C\equiv C(p_-,p_+,q) >0$ such that and all $(x,y)$ under the condition $x\ge C$
$$
(\ref{e1})  \le q < 1.
$$
Thus, in this case we define the set $B$ from the condition as $B=\{ (x,y): \ x\ge C\}$. This complete the proof of the geometrically ergodicity of the chain $\hat{s}_n$.  $\Box$

The second step, we should to prove that $E_{\hat\pi} |F(\hat{s}_n)|^{2+\delta} < \infty$ for some $\delta >0$, where the function $F$ defined by (\ref{f}). For this we need some information of the behavior of the invariant measure.

As before, let $\pi$ be the invariant measure for the chain $s_n$. The condition takes the following form.
\begin{equation}\label{e2}
\begin{aligned}
\mathbb{E}_{\hat{\pi}} |F(\hat{s}_n)|^{2+\delta} & = \sum_{x=1}^\infty \Bigl( \sum_{k=1}^{x-1} k^{2+\delta} \hat{\pi}(x,x-k) \frac{\beta_+}{\gamma_-} + \hat{\pi}(x,x+1) \frac{\beta_-}{\gamma_+} \Bigr) \\
& = \frac{\beta_-}{\gamma_+ + \gamma_-} + \frac{\beta_+}{\gamma_+ + \gamma_-}  \sum_{x=1}^\infty  \pi(x)  \sum_{k=2}^{x-1} \frac{k^{2+\delta}}{x-1}\\
& <  \frac{\beta_-}{\gamma_+ + \gamma_-} + \frac{\beta_+}{\gamma_+ + \gamma_-} \frac{1}{3+\delta}  
\sum_{x=1}^\infty  \pi(x) \frac{x^{3+\delta}}{x-1}
\end{aligned}
\end{equation}
Thus, if we prove that $\pi(x)$ decreases sufficiently, then the last series in (\ref{e2}) will converge. Indeed, 
from the relation (\ref{inv1}) we have directly
%the invariant measure for $s_n$ satisfies the following global balance equation: let as before $p = \frac{\gamma_+}{\gamma_++\gamma_-} $ and $q = \frac{\gamma_-}{\gamma_++\gamma_-}$, then
%$$
%\left\{ 
%\begin{aligned}
%\mu (1) & =  q\mu(2) + \frac{q}{2}\mu(3) + \frac{q}{3}\mu(4) + \dots \\
%\mu (2) & = \mu(1) + \frac{q}{2}\mu(3) + \frac{q}{3}\mu(4) + \dots \\
%\mu (3) & = p \mu(2) + \frac{q}{3}\mu(4) + \frac{q}{4}\mu(5) + \dots \\
%\dots  \\
%& \sum_{i=1}^\infty \mu(i) =1.
%\end{aligned}
%\right.
%\Leftrightarrow
%\left\{ 
%\begin{aligned}
%2 \mu (1) & = (1+q)\mu(2) \\
%(1+p) \mu (2) & = \mu(1) + (1+\frac{q}{2})\mu(3) \\
%(1+p) \mu (3) & = p \mu(2) + (1+\frac{q}{3})\mu(4) \\
%\dots \\
%& \sum_{i=1}^\infty \mu(i) =1
%\end{aligned}
%\right.
%$$
%$$
%\Leftrightarrow
%\left\{ 
%\begin{aligned}
%\mu (2) & = \frac{2}{1+q} \mu(1) \\
%\mu (3) & = \frac{3p}{(1+q)(1+\frac{q}{2})} \mu(1) \\
%\mu (n) & = \frac{1+p}{1+\frac{q}{n-1}} \mu(n-1) - \frac{p}{1+\frac{q}{n-1}} \mu(n-2), \ \ n \ge 4 \\
%& \sum_{i=1}^\infty \mu(i) =1
%\end{aligned}
%\right.
%\Leftrightarrow
%\left\{ 
%\begin{aligned}
%\mu (2) & = \frac{2}{1+q} \mu(1) \\
%\mu (n) & = \frac{np^{n-2}}{ \prod_{i=1}^{n-1} (1+\frac{q}{i})} \mu(1), \ \ n \ge 2 \\
%& \sum_{i=1}^\infty \mu(i) =1.
%\end{aligned}
%\right.
%$$
%Soon, we have
$$
\pi(n) < np^{n-2} \pi(1),
$$
which provide the convergence of the last series in (\ref{e2}). It proves the conditions for CLT.

\begin{lemma} \label{CLT}
$$
\frac{p_n - n v}{\sqrt{ n \mathbb{V}ar_{\hat\pi}\big(F(\hat{s}_n)\big)}} \overset{d}{\longrightarrow} N(0, 1) \hspace{0.5cm} \textrm{a.s.} \hspace{0.5cm} n \to \infty
$$
where
$$
\mathbb{V}ar_{\hat\pi}\big(F(\hat{s}_n)\big)= \frac{\beta_-}{\gamma} + \frac{\mu(1)}{\gamma} \Bigl( \frac{\beta_-\gamma_-}{\gamma_+} - \frac{\beta_+}{6} \Bigr) + \frac{\beta_+}{\gamma} \frac{1}{6}  \sum_{s=1}^\infty s(2s-1)\pi(s)-\big[\mathbb{E}_{\hat\pi} (F(\hat{s}_n))\big]^2
$$
\end{lemma}

The Lemma \ref{CLT} relates the ``coarse-grained'' volatility of intraday returns at lower frequencies to the high-frequency jumps rates of prices. In simple terms, Lemma \ref{CLT} states that, observed over time, the prices has a diffusive behavior around  a local drift with a diffusion coefficient $\mathbb{V}ar_{\hat\pi}\big(F(\hat{s}_n)\big)$. Therefore, price volatility, as a function of the number of micro-jumps in prices, is given by
\begin{equation}\label{diff_v}
\sigma_n = {\sqrt{ n \mathbb{V}ar_{\hat\pi}\big(F(\hat{s}_n)\big)}}
\end{equation}
where $n$  is the total number of high-frequency prices jumps. Formula (\ref{diff_v}) yields an estimator for price volatility which may computed without observing the price on long-term. Optionally, the parameter $\sigma_n$ can be interpreted as the intraday realized volatility of the asset. Therefore, the relation (\ref{diff_v}) links the realized volatility with the high-frequency parameters of the OB.

\subsubsection{CLT for continuous-time $P_b(t)$}

Based on previous result let us prove the CLT for the $P_b(t)$. As before let $N_t$ be the Poisson process with rate $\gamma = \gamma_++\gamma_-$, which is the number of jumps of the process $X(t)$. The following representation takes place.
$$
\begin{aligned}
\sqrt{t} \left( \frac{P_b(t)}{t} - D \right) = \sqrt{N_t}\left( \frac{p_{N_t}}{N_t} - v \right)  \sqrt{\frac{N_t}{t}} + \sqrt{t} \left( \frac{N_t}{t} v - D \right) 
\end{aligned}
$$ 
According Lemma~\ref{CLT} and CLT for the Poisson process we expect
\begin{equation}\label{twodist}
\begin{aligned}
\sqrt{N_t}\left( \frac{p_{N_t}}{N_t} - v \right) & \Rightarrow N(0, \sigma^2) \mbox{ in distribution for some }\sigma^2 \\
\sqrt{t} \left( \frac{N_t}{t} - \gamma \right) & \Rightarrow N(0, \gamma) \mbox{ in distribution} \\
\frac{N_t}{t} & \to \gamma \mbox{ a.s. }
\end{aligned}
\end{equation}
The second and third convergence are well known, when the first one can be proved as follows: let $F_{n,\sigma^2} (x)$ be the cumulative distribution function of the scaled imbedded price Markov chain $p_n$ from Lemma~\ref{CLT}, then
for any $\delta >0$ there exists $n_\delta$ such that for all $n>n_\delta$
$$
\left| F_{n} (x) - \Phi_{\sigma^2} (x) \right| < \delta,
$$
where $\Phi_{\sigma^2} (x)$ stands for the cumulative normal distribution with zero mean and variance $\sigma^2$. Then
$$
\begin{aligned}
\mathbb{P} \left( \sqrt{N_t}\left( \frac{p_{N_t}}{N_t} - v \right) \le x \right) 
& = \mathbb{P}  \left. \left( \sqrt{N_t}\left( \frac{p_{N_t}}{N_t} - v \right) \le x \ \right|\ N_t \le n_\delta \right) 
\mathbb{P} (N_t \le n_\delta) \\
& + \sum_{n=n_\delta + 1}^\infty \mathbb{P}  \left. \left( \sqrt{N_t}\left( \frac{p_{N_t}}{N_t} - v \right) \le x \ \right|\ N_t = n \right) \mathbb{P} (N_t = n) \\
& \le \mathbb{P} (N_t \le n_\delta) + \Phi_{\sigma^2} (x) + \delta \\
& \to \Phi_{\sigma^2} (x) \mbox{ as }t\to\infty, \mbox{ and as }\delta \to 0 \\
\mathbb{P} \left( \sqrt{N_t}\left( \frac{p_{N_t}}{N_t} - v \right) \le x \right) 
& \ge \left( \Phi_{\sigma^2} (x) - \delta \right)  \mathbb{P} (N_t > n_\delta) \\
& = \Phi_{\sigma^2} (x) - \delta  + \left( \Phi_{\sigma^2} (x) - \delta \right)  \mathbb{P} (N_t \le n_\delta)  \\
& \to \Phi_{\sigma^2} (x) \mbox{ as }t\to\infty, \mbox{ and as }\delta \to 0
\end{aligned} 
$$
observe that two normal variables from (\ref{twodist}) are not independent, but they are asymptotically non-correlated, moreover, they are asymptotically independent. Indeed, let us show that 
$$
\sqrt{N_t}\left(\frac{p_{N_t}}{N_t}-v\right) \text{ and } \sqrt{t}\left(\frac{N_t}{t}-\gamma\right)
$$
are asymptotically independent. Since the second sequence is the measurable function of the $N_t$ it is enough to prove that for all $x\in \mathbb{R}$, set $A\subseteq \mathbb{N}$ and $\varepsilon>0$ the following inequality holds
\begin{equation}\label{as.ind}
\lim\limits_{t\rightarrow\infty}\left|\mathbf{P}\left(\sqrt{N_t}\left(\frac{p_{N_t}}{N_t}-v\right)<x,N_t\in A\right)-
\Phi_{\sigma^2}(x)\mathbf{P}(N_t\in A)\right|\leq\varepsilon.
\end{equation}
If the set $A$ is bounded from above, then the inequality holds:
$$
0\leq\lim\limits_{t\rightarrow\infty}\mathbf{P}\left(\sqrt{N_t}\left(\frac{p_{N_t}}{N_t}-v\right)<x,N_t\in A\right)\leq
\lim\limits_{t\rightarrow\infty}\mathbf{P}(N_t\in A)=0.
$$
Suppose now that $A$ is not bounded. In this case for any $\delta>0$ we obtain
$$
\begin{aligned}
&\lim\limits_{t\rightarrow\infty} \mathbf{P}\left(\sqrt{N_t}\left(\frac{p_{N_t}}{N_t}-v\right)<x,N_t\in A\right) \\
& =
\lim\limits_{t\rightarrow\infty}\mathbf{P}\left(\sqrt{N_t}\left(\frac{p_{N_t}}{N_t}-v\right)<x,N_t\in A\cap[0,n_\delta]\right)
+\lim\limits_{t\rightarrow\infty}\sum\limits_{k=n_\delta+1}^\infty\mathbf{P}\left(\sqrt{k}\left(\frac{p_{k}}{k}-v\right)<x,
N_t\in A\cap\{k\}\right)
\\ &=\lim\limits_{t\rightarrow\infty}\sum\limits_{k=n_\delta+1}^\infty\mathbf{P}\left(\sqrt{k}\left(\frac{p_{k}}{k}-v\right)<x\right)
\mathbf{P}(N_t\in A\cap\{k\}).
\end{aligned}
$$
Thus, for any $\delta>0$ we have
$$
\begin{aligned}
\lim\limits_{t\rightarrow\infty}(\Phi_{\sigma^2}(x)-\delta)\mathbf{P}(N_t\in A\cap[n_\delta,\infty))& \leq\lim\limits_{t\rightarrow\infty}\mathbf{P}\left(\sqrt{N_t}\left(\frac{p_{N_t}}{N_t}-v\right)<x,N_t\in A\right)
\\ &\leq
\lim\limits_{t\rightarrow\infty}(\Phi_{\sigma^2}(x)+\delta)\mathbf{P}(N_t\in A\cap[n_\delta,\infty)),
\end{aligned}
$$
and the following inequality holds
$$
\lim\limits_{t\rightarrow\infty}\left|\mathbf{P}\left(\sqrt{N_t}\left(\frac{p_{N_t}}{N_t}-v\right)<x,N_t\in A\right)-
\Phi_{\sigma^2}(x)\mathbf{P}(N_t\in A)\right|
\leq\delta\lim\limits_{t\rightarrow\infty}\mathbf{P}(N_t\in A\cap[n_\delta,\infty))<\delta.
$$
Choosing $\delta=\varepsilon$, we obtain the inequality (\ref{as.ind}).

%Indeed, let $a_n = \mathbb{E} \left( \sqrt{n}\left( \frac{p_n}{n} -  v \right)\right)$, then we know that $a_n \to 0$ as $n\to\infty$. The expectation of the product can be represented as
%$$
%\begin{aligned}
%v & \mathbb{E} \left( \sqrt{N_t}\left( \frac{p_{N_t}}{ N_t} -  v \right) \sqrt{ \frac{N_t}{t}} \sqrt{t}\left( \frac{N_t}{t} -\lambda \right) \right) = v  \mathbb{E} \left( a_{N_t} \sqrt{ \frac{N_t}{t}} \sqrt{t}\left( \frac{N_t}{t} -\lambda \right)\right) \\
%& \le v  \sqrt{ \mathbb{E} (a_{N_t}^2) }  \sqrt{ \mathbb{E} \left( N_t \left(   \frac{N_t}{t} -\lambda  \right)^2 \right) }  =
%v  \sqrt{ \mathbb{E} (a_{N_t}^2) } \left( \gamma^2 + \frac{\gamma}{t} \right) \\
%& = \mu \left( \frac{cov(p_{N_t}, N_t) }{t} - \lambda \right) \to cov.
%\end{aligned}
%$$
%In the same way as before, it is easy to see, that $\mathbb{E} (a_{N_t}^2) \to 0$, as $t\to\infty$. 

It finishes the proof of the convergence
$$
\begin{aligned}
\sqrt{t} \left( \frac{P_b(t)}{t} - D \right) \to N\left(0, (\sigma^2 + v^2)\gamma\right)
\end{aligned}
$$
$\Box$
%%--------------------------------------------------

%\noindent
\subsection{Large Deviations for the spread $S(t)$.} 

It is known that in liquidity fluctuations contexts even a small order can create a large price change and consequently create a very large spread (\cite{doyne2004really}, \cite{bouchaud2009markets}). Therefore, our interest is to understand how large
changes in the spread occur without altering parameters of the model. We believe that this type of analysis can be used to assess the resilience of the OB to severe fluctuations in the liquidity. 

In this section we present an application of the large deviations theory to Markov process describing dynamics of the spread, i.e., we study the large deviations asymptotics for spread process. Our goals is to find the most probable trajectory corresponding to a certain state of spread, in particular very large, during the time interval.

Large deviations for the Poisson processes with uniform (almost uniform) catastrophes was considered in the papers \cite{logachov2019large} and \cite{logachov2018local}. Large deviation can be considered as a finishing step in a sequence of limit theorems for the processes. The theory of large deviation is well developed at the moment, but the processes considered here do not satisfy the ``classical'' conditions, because why the proof of the large deviations is still very technic. 

In order to provide the large deviations we need some increasing scaling parameter. Let $T$ be the length of the time interval $[0,T]$ we observed our process. We consider the following scaled process
$$
S_T(t) = \frac{S(tT)}{T}, \ \ t\in [0,1].
$$ 
We say that the family of random variable $S_T(1)$ satisfies large deviation principle (LDP) on $\mathbb{R}$ with the rate function $I=I(x): \mathbb{R} \to [0,\infty]$, if for any $c\ge 0$ the set $\{x\in\mathbb{R}: \ I(x) \le c\}$ is compact and for any set $B\in \mathcal{B}(\mathbb{R})$ the following inequalities hold:
$$
\limsup_{T\to\infty} \frac{1}{T} \ln \mathbb{P}(S_T(1) \in B) \le - \inf_{x\in [B]} I(x)\ \  \mbox{ and }\ \ 
\liminf_{T\to\infty} \frac{1}{T} \ln \mathbb{P}(S_T(1) \in B) \le - \inf_{x\in (B)} I(x),
$$
where $\mathcal{B}(\mathbb{R})$ is the Borel $\sigma$-algebra on $\mathbb{R}$ and $[B], (B)$ are closure and open interior of the set $B$ correspondingly. This principle was established in the paper \cite{logachov2019large}, in which the logarithmic asymptotic for the probability $\mathbb{P}(S_T(1) >x)$ was calculated. Note, that the principle was proved for the state $x$ of the spread at the time $T$, it is not the principle on the functional space. The principle on the functional (trajectory) space provides us the possibility to find the (unique) \textit{optimal} trajectory -- the trajectory which shows how such deviation (rare event) occurs taking into account the evolution of the spread.

As a first attempt for proving the principle on the functional space is to prove the \textit{local} large deviation -- the asymptotics for the probability of the process stay in a small neighborhood of some continuous function. We say that the family of the processes $S_T(\cdot)$ satisfies \textit{local large deviation principle} (LLDP) on the set $G\subset \mathbb{D}[0,1]$ with rate function $I=I(f): \mathbb{D}[0,1] \to [0,\infty]$ if for any function $f\in G$ the following inequalities hold
$$
\lim_{\varepsilon\to 0}\limsup_{T\to\infty} \frac{1}{T} \ln \mathbb{P}(S_T(\cdot) \in U_\varepsilon (f)) =
\lim_{\varepsilon\to 0} \liminf_{T\to\infty} \frac{1}{T} \ln \mathbb{P}(S_T(\cdot) \in U_\varepsilon (f)) = - I(f),
$$
where $\mathbb{D}[0,1]$ is the space of c\`adl\`ag functions, i.e. the functions that are continuous from the right, and have a limit from the left; and where $U_\varepsilon (f) := \{ g\in \mathbb{D}[0,1]: \sup_{t\in[0,1]} |f(t)-g(t)| \le \varepsilon\}$. 

The LLDP was proved in \cite{logachov2018local} for the compound Poisson process with almost uniform catastrophes. We note only here that the process $S(\cdot)$ is the special case of the processes considered in \cite{logachov2018local}. In order to write the corresponding rate function we need to remind that any function with the finite variation can be uniquely represented as a difference of two nondecreasing functions $f^+$ and $f^-$ such that ${\rm Var} f_{[0,1]} = {\rm Var} f^+_{[0,1]} + {\rm Var} f^-_{[0,1]}$. The functions $f^+$ and $f^-$ are called the positive and negative variations of the function $f$ respectively. Now, the rate function for $S_T(\cdot)$  can be represented as follows
\begin{equation}\label{RF}
I(f) = \gamma_- + \int_0^1 \left( \dot{f}^+ (t) \ln \Bigl( \frac{\dot{f}^+(t)}{\gamma_ +} \Bigr) - \gamma_+ \Bigl( \frac{\dot{f}^+(t)}{\gamma_+} -1\Bigr) \right) \mathbf{I} \left(\dot{f}^+ (t) > \gamma_+\right) dt,
\end{equation}
where $\dot{f}$ stands for the derivative of function $f$ and $\mathbf{I}$ is the indicator function. 

	\begin{figure}[h!]
		\centering
	           \includegraphics[scale=0.55]{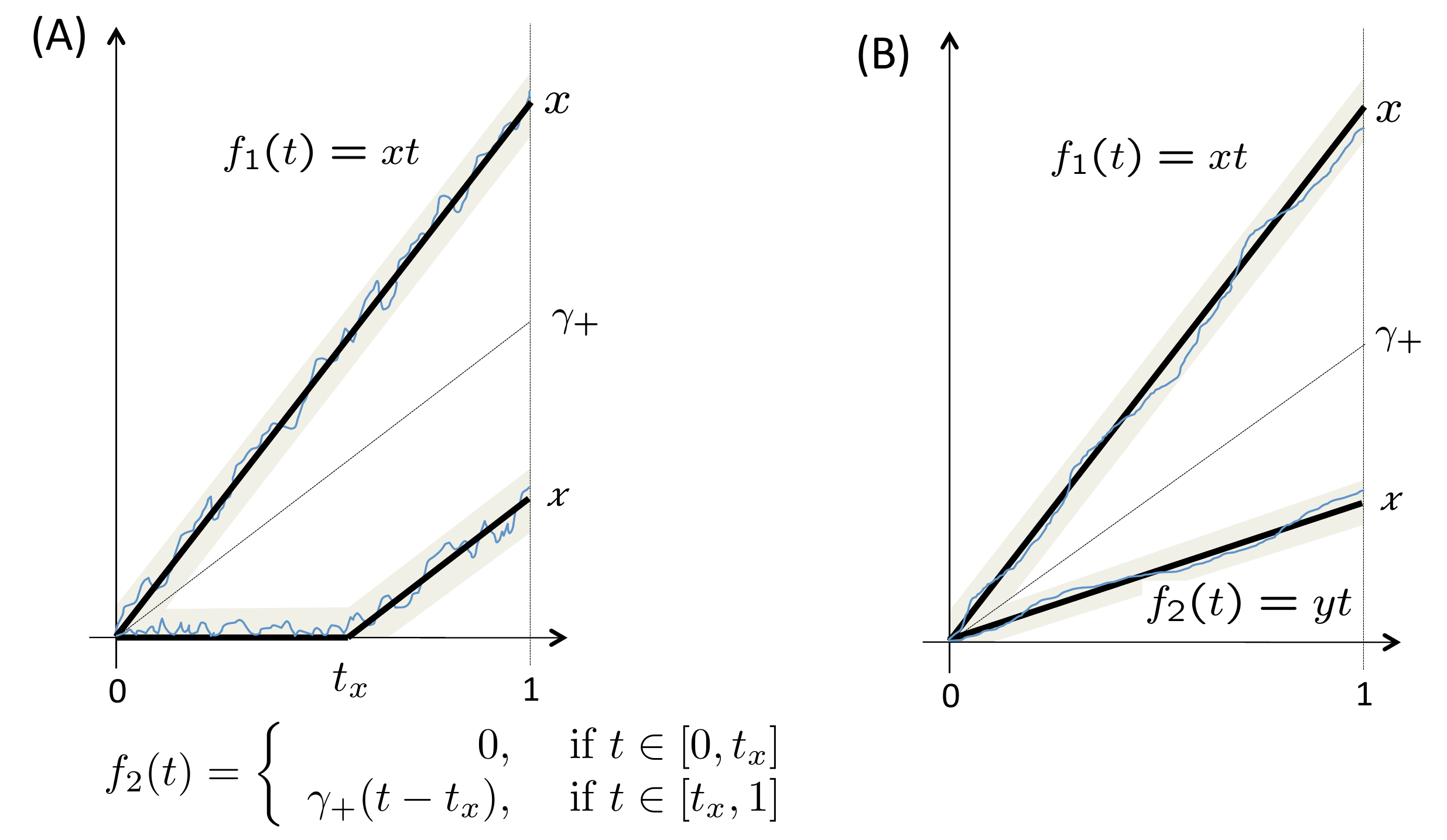} 
		 \caption{The optimal trajectories for (A) Spread process, which is a Poisson process (of rate $\gamma_+$) with uniform catastrophes (of rate $\gamma_-$), and (B) Poisson process with rate $\gamma_+$. If $x<\gamma_+$ then the large deviation occurs according to the functions $f_2$. If $x \geq \gamma_+$ then the large deviation trajectory is in the neighborhood of the straight line $f_1$.}
                      \label{Optimal}
	\end{figure}  

We note that the both large deviation principle and local large deviation have the same normalization factor for the probabilities, $1/T$. It provides the existence of an \textit{optimal} trajectory for the large deviations. The existence of the optimal trajectories of large deviations $S_T(1) >x$ were established in \cite{logachov2019large}. If $x<\gamma_+$ then there exists the moment $t_{x}=1-\frac{x}{\gamma_+} \in (0,1)$ such that the spread process $S_T(\cdot)$ stays near the zero up to the time $t_{x}$ and after that $S_T(t)$, $t\geq t_{x}$ increases according to the straight line which starts at point $(t_{x},0)$ and grows up to the point $(1,x)$ with the slope $\gamma_+$, see the function $f_2$ on Fig. \ref{Optimal}-(A). If $x\geq \gamma_+$ then the process grows together with the straight line starting from origin up to the point $(1,x)$, i.e., its slope $x$, the function $f_1$ on Fig. \ref{Optimal}-(A). For illustrative purposes of comparison, in the Fig.~\ref{Optimal} we represent the optimal trajectories which provide large fluctuations for the Poisson process with rate $\gamma_+$ and the process $S_T$, that is the Poisson process (of rate $\gamma_+$) with uniform catastrophes (of rate $\gamma_-$). 
%%--------------------------------------------------

\subsection{Large Deviations for the prices $(P_b(t), P_a(t))$} 

The large deviation result for spread suggests the question about the behavior of prices under the large spread. The rate function corresponding the large deviation is essentially the rate function of a Poisson process with rate $\gamma_+$ which consists of the rates $\alpha_+$ and $\beta_-$. Here we provide some qualitative behavior of optimal price trajectories without the proof. The qualitative picture is represented in Figure~\ref{OptimalP}. 

The main difference between the behavior of optimal trajectories of Poisson processes and our process is the presence of the indicator in the rate function, see (\ref{RF}). The indicator restricts the values of line slope -- it cannot be lesser then the rate of Poisson process. Thus, when the scaled spread is lesser then $\gamma_+$ there exists the ``bifurcation'' point $t_x$, and after that the slopes of two lines are $\alpha_+$ for the upper line and $-\beta_+$ for the lower one. The slopes changes when the scaled spread is greater then $\gamma_+$, but the relation between contributions of rates $\alpha_+$ and $-\beta_+$ is the same.

\begin{figure}[h!]
		\centering
	           \includegraphics[scale=0.55]{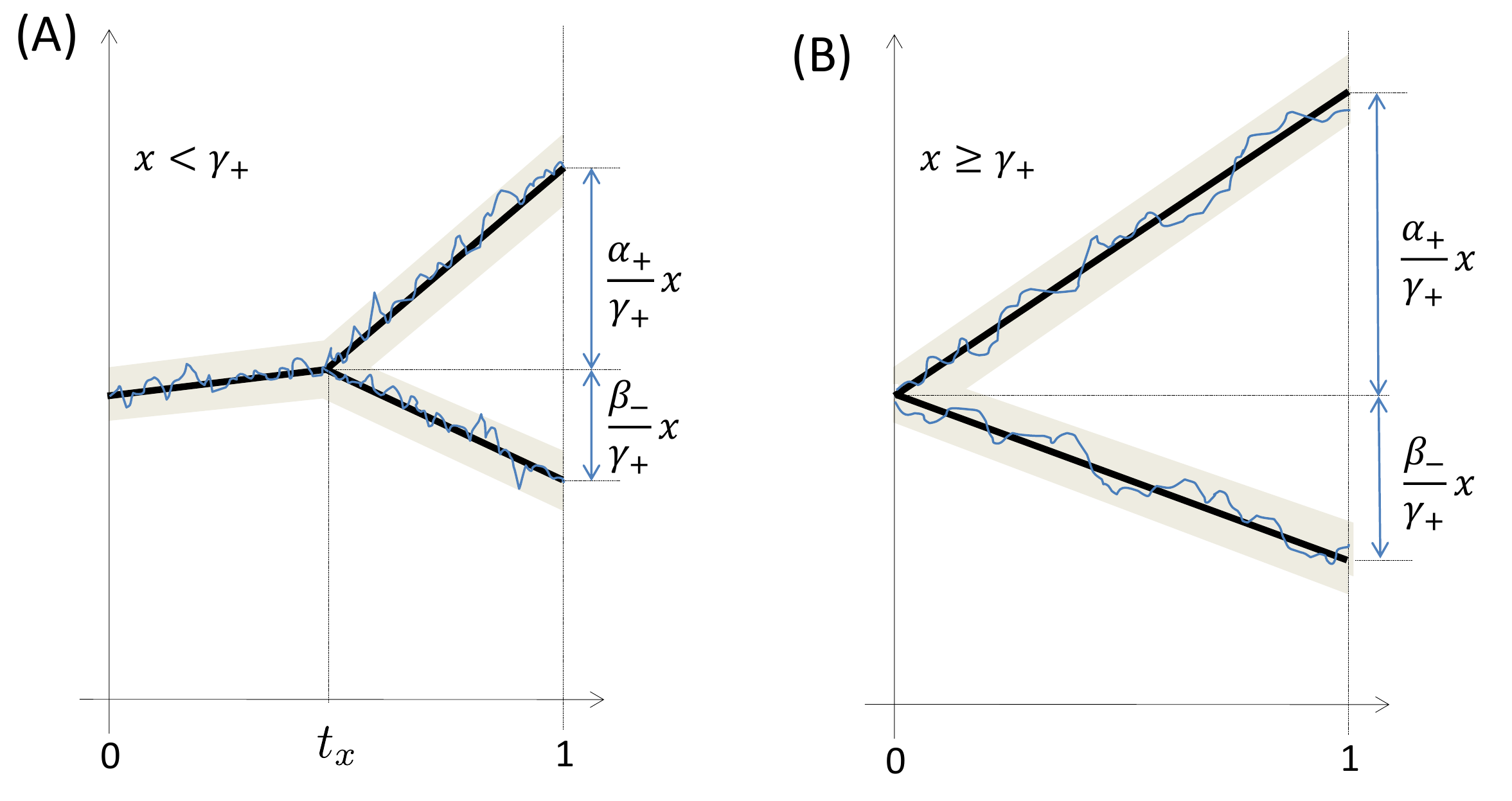} 
		 \caption{The optimal trajectories for prices $(P_b(t), P_a(t))$ under the large deviation of the spread when (A) the scaled spread $x$ is lesser then $\gamma_+$, which is consists of the rates $\beta_-, \alpha_+$, i.e.  $\gamma_+= \beta_- + \alpha_+$; and (B) the scaled spread $x \ge \gamma_+$.}
                      \label{OptimalP}
	\end{figure}

\subsection{Numerical Results}

In this section, we explore steady state properties of our proposed model using Monte Carlo simulation. We compare empirically observed long-term behavior (unconditional properties) of the OB to simulations of the fitted model. The goal of these simulations is to indicate how well the model reproduces the average properties of the OB. The transitions rates of $X(t)$ can be estimated by 

\begin{equation*}
\hat{\alpha}_+=\frac{N_{\alpha_+}}{T}, \quad  \hat{\alpha}_-=\frac{N_{\alpha_-}}{T}, \quad  \hat{\beta}_+=\frac{N_{\beta_+}}{T}, \quad \hat{\beta}_-=\frac{N_{\beta_-}}{T}.
\end{equation*}

where $T$ is the length of our sample (in seconds),  $N_{\alpha_+}$ ($N_{\alpha_-}$)  is the total number of jumps where the ask price increases (decreases) and $N_{\beta_+}$ ($N_{\beta_-}$)  is the total number of jumps where the bid price increases (decreases). The fitted values for Apple  Inc. stock, for a 15 minute sample, are: $\hat{\alpha}_+=5$, $\hat{\alpha}_-=3$, $\hat{\beta}_+=2$ and $\hat{\beta}_-=4$. Based on the estimation of parameters $(\hat{\alpha}_+, \hat{\alpha}_-,  \hat{\beta}_+,  \hat{\beta}_-)$, we simulate the price process $X(t)$ over a long horizon of $900$ seconds, which corresponds to what was empirically observed, and observe the evolution of prices in two time windows. The results are displayed in Figure \ref{Simul_prices}. The results of our simulations illustrate that our model reproduces realistic characteristics for price behavior, both short and long term, which were presented for empirical data in Figures \ref{short_long_prices}. % and \ref{PricesSLR}.

	\begin{figure}[h!]
		\centering
		\begin{subfigure}[b]{0.35\linewidth}
			\includegraphics[width=\linewidth]{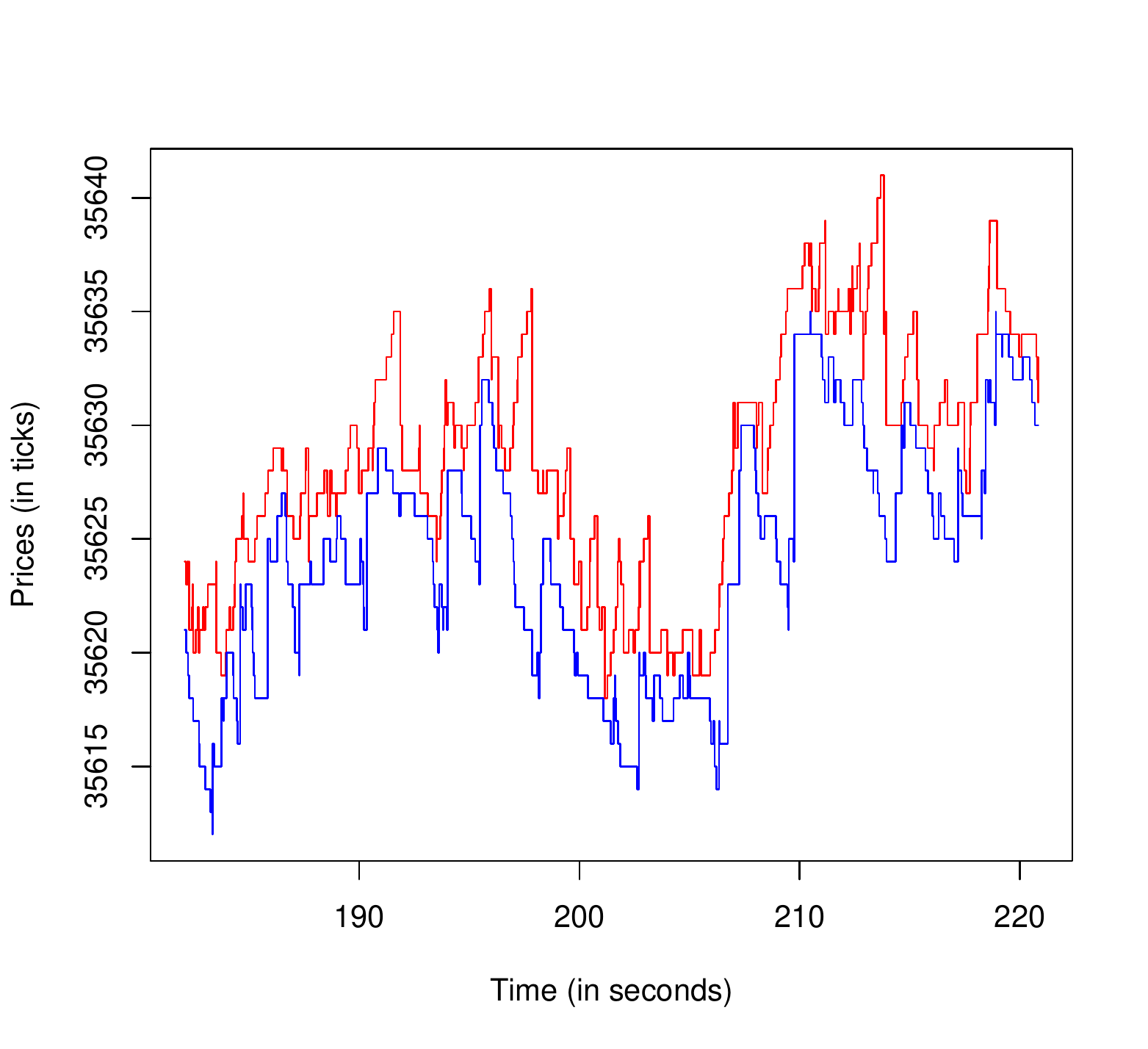} 
		\end{subfigure}
		\begin{subfigure}[b]{0.35\linewidth}
			\includegraphics[width=\linewidth]{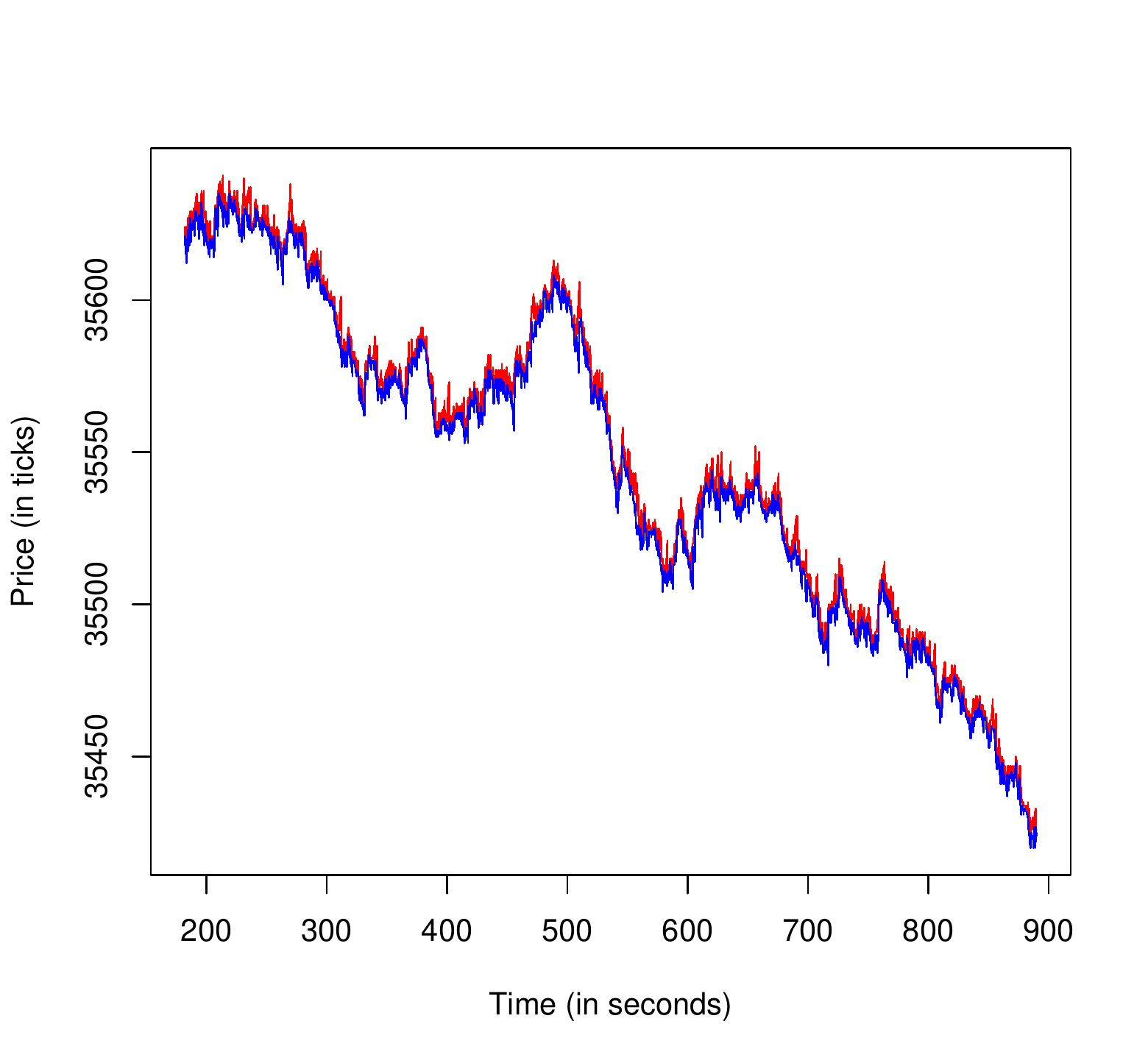}
		\end{subfigure}
		\begin{subfigure}[b]{0.35\linewidth}
			\includegraphics[width=\linewidth]{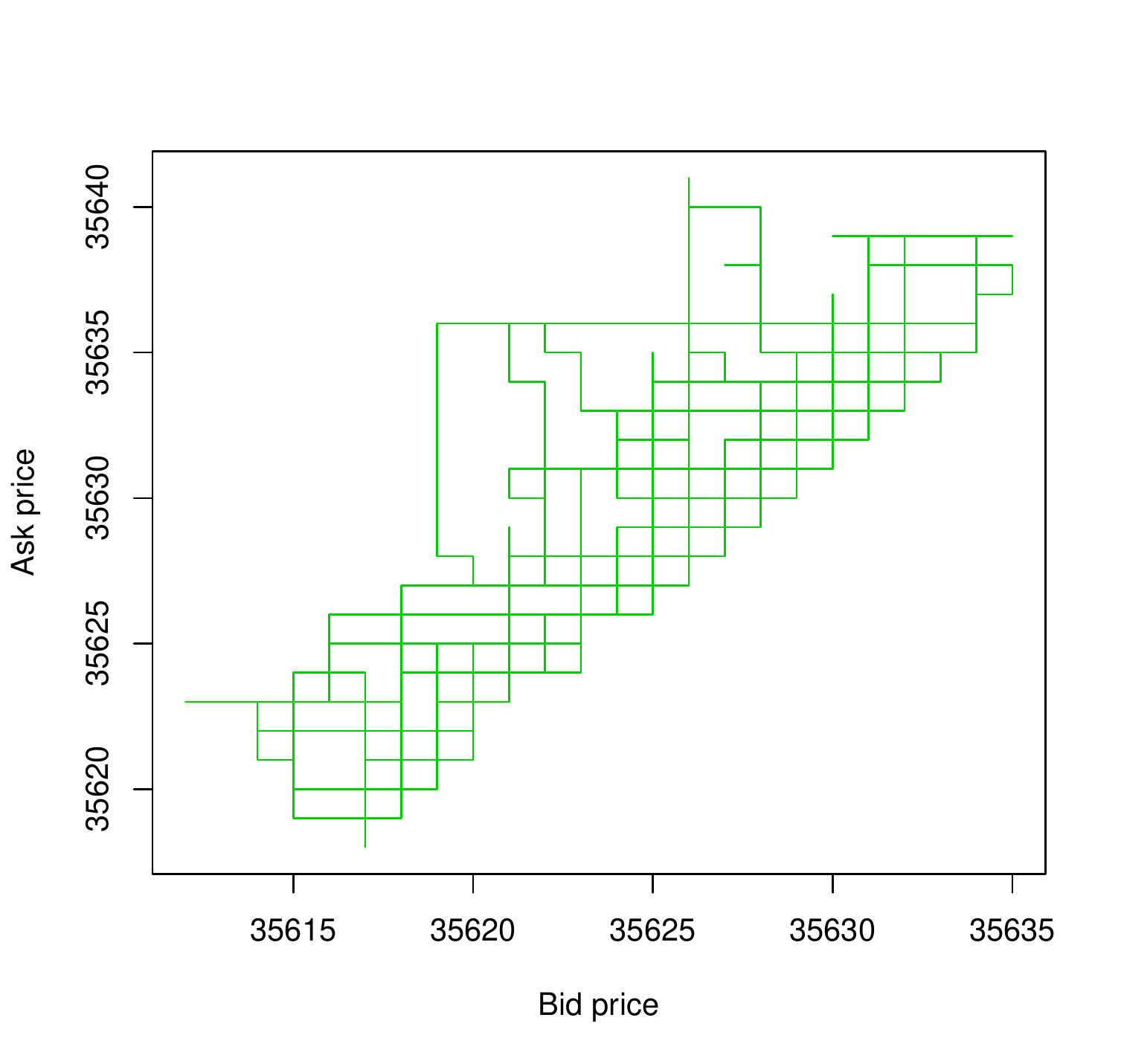} 
		\end{subfigure}
		\begin{subfigure}[b]{0.35\linewidth}
			\includegraphics[width=\linewidth]{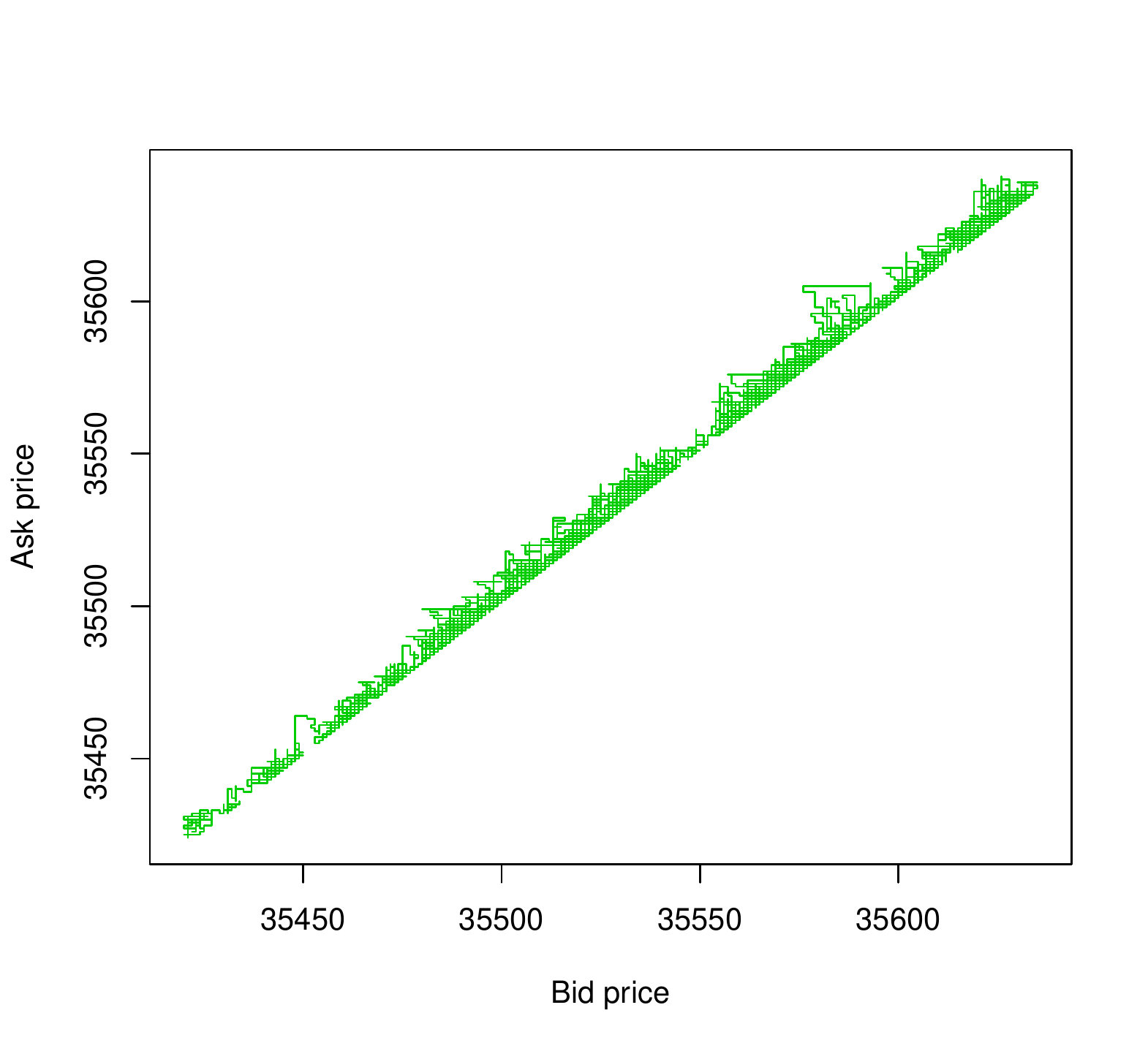}
		\end{subfigure}
    \caption{ Simulation of the order book with parameters: $\hat{\alpha}_+=5$, $\hat{\alpha}_-=3$, $\hat{\beta}_+=2$ and $\hat{\beta}_-=4$. Upper left: Short-term evolution of bid (blue) and ask (red) prices, 1 minute sample. Upper right: Long-term evolution of the prices, 15 minute sample. Bottom left: Short-term path of the price process $X(t)$, 1 minute sample. Bottom right: Long-term path of the price process, 15 minute sample. }
     \label{Simul_prices}
	\end{figure}  
	
The simulation results illustrate that our model also yields  realistic features for (steady state) average behavior of the OB profile, within which we can highlight the negative autocorrelation of price changes at first lag. It is empirically observed that high frequency price movements have a negative autocorrelation at the first lag. The autocorrelation function of transaction price returns is strongly negative at the first lag and then it rapidly decreases to zero, see Figure \ref{acf}. This is the well-known bid-ask bounce (\cite{roll1984simple}) and is mainly due to the presence of two trading prices, one for buyer and one for seller initiated transactions. This negative autocorrelation disappears when one considers aggregate returns, and it is therefore a typical microstructural effects that it is important that it be considered in a model for the OB. Our model satisfactorily reproduces this empirical fact.

	\begin{figure}[h!]
		\centering
	           \includegraphics[scale=0.50]{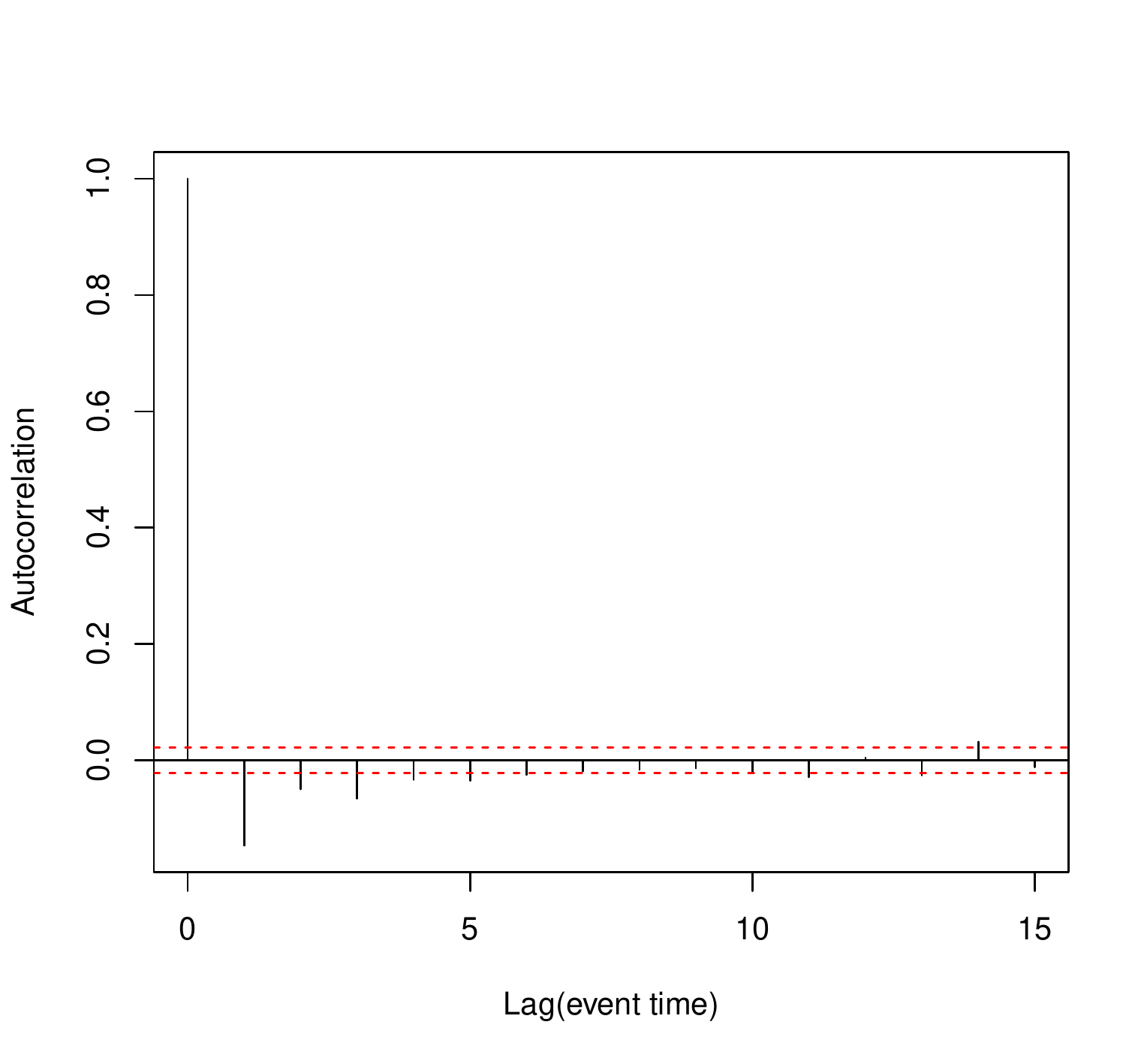} 
		 \caption{Autocorrelation function of price return based our simulations of the order book.}
                      \label{acf}
	\end{figure}  
	
Finally, a quantity particularly important for traders who are interested in trading in a short time scale is the probability that the price will increase ar the next move given a state of the OB. Based on our model this probability is very simple to calculate and is given by

\begin{equation*}
\mathbb{P}\big(\Delta P >0 \mid OB\big)=\frac{\alpha_++\beta_+}{\gamma_++\gamma_-}
\end{equation*}
where $\Delta P$ defined to be the next midprice move, and for OB defined to be price process $X(t)$ and their respective transition rates. It is expected that predictions based on this quantity will have a better long-term performance. Since in trading the interest is very short-term predictions, it would be interesting to include the bid and ask queue size in our model, this may be a subject of future research.

%%--------------------------------------------------

\section{Discussion. Other regimes}\label{other_regimens}

In this section we formulate the other two regimes that can be considered within our general model. As previously mentioned, the main results of this article can be generalized for these other two regimes, although we believe that qualitatively there will be no significant differences between the results. We hope that these two formulations that follow will motivate future works that generalize our results.

\subsection{Almost uniform catastrophes}\label{almostUC}

As we mentioned before the large deviations were proved for the so-called \textit{almost uniform} catastrophes. Recall, that in order to close the spread of length $k$ (with probability $\frac{\gamma_-}{\gamma}$) we choose the next state for the spread with the same probability (uniformly) from the set $I_k=\{1, \dots, k-1\}$, denote these probabilities as $Q_i(k), i\in I_k$, and here $Q_i(k)=\frac{1}{k-1}$. The \textit{almost-uniform} distribution is defined by the following form of probabilities $Q_i(k), 1\le i \le k-1$: there exists a constant $c>1$ such that for all $k\in\mathbb{N}$
$$
\frac{1}{c (k-1)} \le Q_i(k) \le \frac{c}{k-1},
$$
for all $i\in I_k$. It extends the class of models for high-competitive regimes. For example, for any length $k$ of the spread, it can be divided into some parts, say two part: and we say that with probability $0.7$ we choose the one part and then uniformly the state from this part is chosen, with probability $0.3$ the second part is chosen and corresponding state we choose uniformly.      

All the proofs above can be slightly modified.  

\subsection{Non-competitive regime}

The main features of the Non-competitive regime (NC Regime) is small openings of the spread, as in HC regime, due to the absence of gaps in the OB, but slow decreasing (power law) of the spread, because the agents that place the limit orders within the spread prioritize an optimal price in their quotes. Compared to the HC regime the agents are less impatient. With some constant rate, the spread open by one tick. For closing spread, let $k>1$ be the spread size, the variation in the prices $\Delta$ is chosen from $I_k=\{1, \dots, k-1\}$ according the rate which is proportional to $\Delta^{-\mu}$ where $\mu$ is a fixed positive number. The parameter $\mu$ can be interpreted as a behavioral measure for agents to obtain a more favorable price in their negotiations.

\subsubsection*{Model: closing the spread polynomially}

In order to define the rates for NC regime, let us again fix parameters $\alpha_{+}^c, \alpha_{-}^c, \beta_{+}^c, \beta_{-}^c,$ which are strictly positive real numbers. 
Suppose that at some moment the chain is at some state $(b,a)\in \mathbb{X}$, then the transition rates for the Markov chain $X(t)$ in this regime are defined by the following way

\begin{equation}\label{HLmodel}
\begin{split}
\alpha_+(\Delta) = \left\{ \begin{array}{rl} \alpha_{+}^c, & \mbox{ if }\Delta =1;\\ 0, & \mbox{ otherwise}; \end{array}\right.
\ \ 
\alpha_-(\Delta) = \left\{ \begin{array}{rl} \frac{\alpha_{-}^c}{\Delta^\mu}, & \mbox{ if } a-b >1\mbox{ for any }\Delta \in I_{a-b}; \\ 0, & \mbox{ otherwise}; \end{array}\right.
\\ 
\beta_-(\Delta) = \left\{ \begin{array}{rl} \beta_{-}^c, & \mbox{ if }\Delta =1;\\ 0, & \mbox{ otherwise}; \end{array}\right.
\ \
\beta_+(\Delta) = \left\{ \begin{array}{rl} \frac{\beta_{+}^c}{\Delta^\mu}, & \mbox{ if } a-b >1\mbox{ for any }\Delta \in I_{a-b}; \\ 0, & \mbox{ otherwise}. \end{array}\right.
\end{split}
\end{equation}

For illustration see Figure \ref{HLNCmodelFig} in the case when $a-b = 3$.

\begin{figure}[h!]
\centering
\includegraphics[scale=0.6]{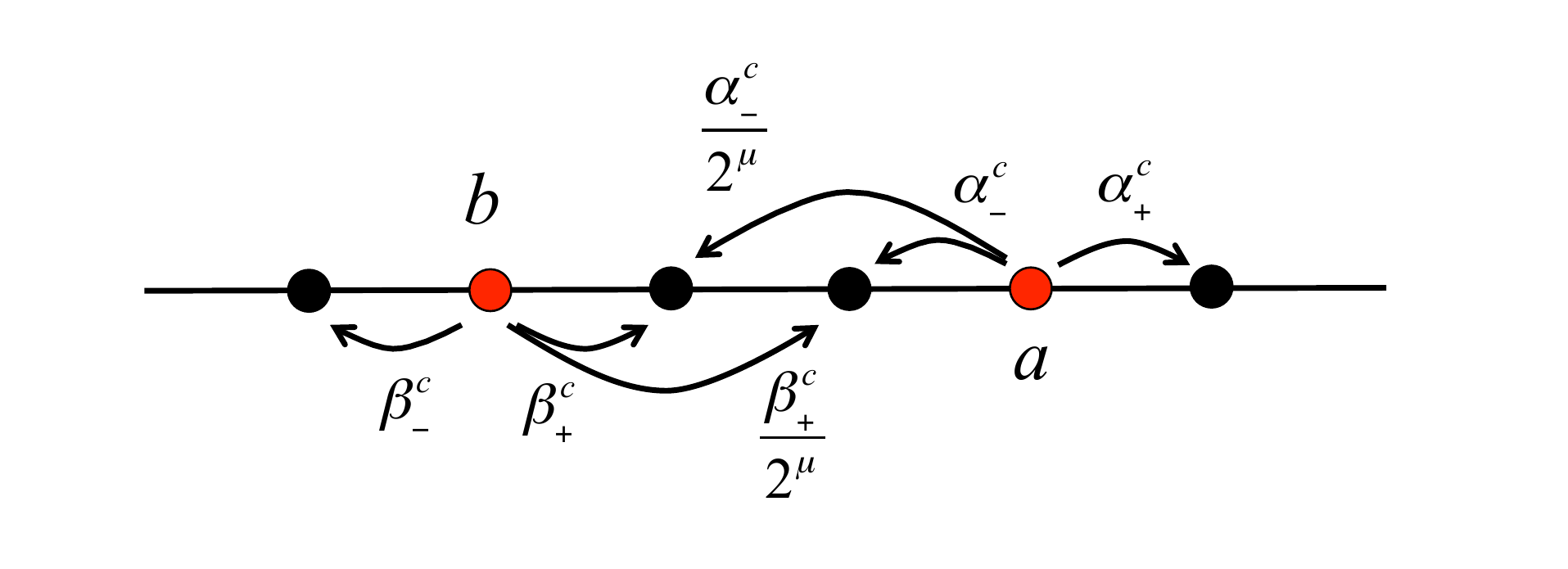}
\caption{The rates for Non-Competitive model. Illustrative example for the case when $a-b = 3$.}
\label{HLNCmodelFig}
\end{figure}

Again, as before, first we study the Markov chain $S(t)$. Suppose that at some moment $t$ the chain is at some state $k \in \mathbb{N}$, and let $\gamma_+^c:=\beta_{-}^c+\alpha_{+}^c$ and $\gamma_-^c:=\beta_+^c + \alpha_{-}^c$, then 

\begin{equation}\label{HLNCmodelS}
\begin{aligned}
k & \to k+1 & \mbox{ with rate } & \ \ \gamma_+^c,\\ 
k & \to k-\Delta & \mbox{ with rate } & \ \ \frac{\gamma_-^c}{\Delta^\mu}  \textrm{ for }  \Delta \in I_{k}.
\end{aligned}
\end{equation}
These transitions suggest that the spread dynamics have a slow reversal process to their typical values, this is because each liquidity provider competes with the others to spread closing.

\subsection{Low liquidity with gaps regime}

The main feature of low liquidity with gaps regime (LLG regime) is that the spread can open by more than one tick, this is due to the existence of gaps in the OB. The spread decreases similarly as NC regime. 

\subsubsection*{Model}

Let us fix parameters $\alpha_{+}^l, \alpha_{-}^l, \beta_{+}^l, \beta_{-}^l, \kappa_a, \kappa_b$ and $\theta\in(0,1)$, which are strictly positive real numbers. Suppose that at some moment the chain is at some state $(b,a)\in \mathbb{X}$, then the transition rates for the Markov chain $X(t)$ for this regime are defined as follows.
\begin{equation}\label{LLmodel}
\begin{aligned}
\alpha_+(\Delta) &= \left\{ \begin{array}{cl} \frac{\alpha_{+}^l}{(a-b)^{\kappa_a} } \cdot \theta(1-\theta)^{\Delta-1}, & \mbox{ for }\Delta \in \mathbb{N};\\ 0, & \mbox{ otherwise}; \end{array}\right.
\\ 
\alpha_-(\Delta) &= \left\{ \begin{array}{cl} \frac{\alpha_{-}^l}{(a-b)^{\kappa_a} } \cdot \frac{\theta(1-\theta)^{\Delta-1}}{1-(1-\theta)^{a-b-1}}, & \mbox{ if } a-b >1\mbox{ for }\Delta \in I_{a-b}; \\ 0, & \mbox{ otherwise}; \end{array}\right.
\\ 
\beta_-(\Delta) &= \left\{ \begin{array}{cl} \frac{\beta_{-}^l}{(a-b)^{\kappa_b} } \cdot \theta(1-\theta)^{\Delta-1}, & \mbox{ for }\Delta \in \mathbb{N};\\ 0, & \mbox{ otherwise}; \end{array}\right.
\\
\beta_+(\Delta) &= \left\{ \begin{array}{cl} \frac{\beta_{-}^l}{(a-b)^{\kappa_b} } \cdot \frac{\theta(1-\theta)^{\Delta-1}}{1-(1-\theta)^{a-b-1}}, & \mbox{ if } a-b >1\mbox{ for }\Delta \in I_{a-b}; \\ 0, & \mbox{ otherwise}. \end{array}\right.
\end{aligned}
\end{equation}
For illustration see Figure~\ref{LLmodelFig} in the case when $a-b = 3$.
\begin{figure}[!htb]
\centering
\includegraphics[scale=0.6]{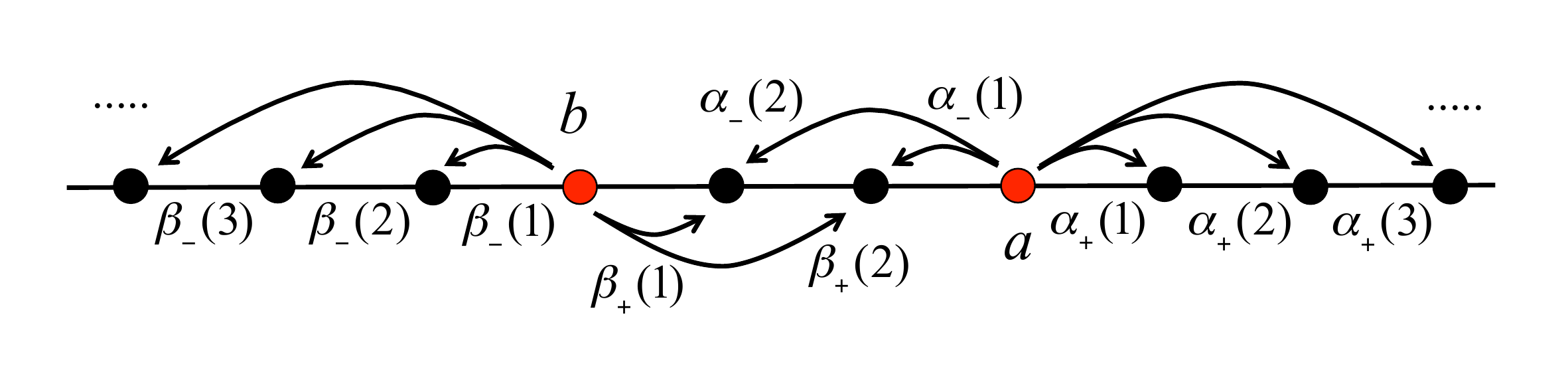}
\caption{The rates for Low Liquidity model. Illustrative example for the case when $a-b = 3$.}
\label{LLmodelFig}
\end{figure}

This Markov chain can be described more easily informally in the following way: for a given state $(b,a)$
\begin{enumerate}
\item[] with the rate $\frac{\alpha_+^l}{(a-b)^{\kappa_a}}$ the chain decides to increase the ask price, and it choose the increment according the geometric distribution with parameter $\theta$;

\item[] with the rate $\frac{\alpha_-^l}{(a-b)^{\kappa_a}}$ the chain decides to decrease the ask price, and it choose the increment according the truncated geometric distribution with parameter $\theta$ and values $I_{a-b} = \{1, \dots, a-b-1\}$;

\item[] with the rate $\frac{\beta_-^l}{(a-b)^{\kappa_b}}$ the chain decides to decrease the bid price, and it choose the increment according the geometric distribution with parameter $\theta$;

\item[] with the rate $\frac{\beta_+^l}{(a-b)^{\kappa_b}}$ the chain decides to increase the bid price, and it choose the increment according the truncated geometric distribution with parameter $\theta$ and values $I_{a-b} = \{1, \dots, a-b-1\}$.
\end{enumerate}

Again, as before, first we study the Markov chain $S(t)$. Suppose that at some moment $t$ the chain is at some state $k \in \mathbb{N}^*$
%, and let $\gamma_+^c:=\beta_{-}^c+\alpha_{+}^c$ and $\gamma_-^c:=\beta_+^c + \alpha_{-}^c$, then 
\begin{equation}\label{LLmodelS}
\begin{aligned}
k & \to k+\Delta & \mbox{ with rate } & \ \ \Bigl( \frac{\alpha_{+}^l}{k^{\kappa_a} } + \frac{\beta_{-}^l}{k^{\kappa_b} } \Bigr)  \cdot  \theta(1-\theta)^{\Delta-1},\mbox{ for }\Delta \in \mathbb{N},\\ 
k & \to k-\Delta & \mbox{ with rate } & \ \ \Bigl( \frac{\alpha_{-}^l}{k^{\kappa_a} } + \frac{\beta_{+}^l}{k^{\kappa_b} } \Bigr) \cdot \frac{\theta(1-\theta)^{\Delta-1}}{1-(1-\theta)^{k-1}}, \textrm{ for }  \Delta \in I_{k}.
\end{aligned}
\end{equation}

%%--------------------------------------------------

  %  \section{Conclusions}

%%--------------------------------------------------
\section*{Acknowledgments}

We thank Vadim Scherbakov for fruitful discussions. Furthermore,  Anatoly Yambartsev thanks the support of FAPESP via grant 2017/10555-0.	

%%--------------------------------------------------
	%------------------------------- Bibliografia
	
	%\bibliographystyle{plainnat}
	\bibliographystyle{apa}
	\bibliography{references}

\begin{thebibliography}{}

\bibitem[\protect\astroncite{Avellaneda
  et~al.}{2011}]{avellaneda2011forecasting}
Avellaneda, M., Reed, J., and Stoikov, S. (2011).
\newblock Forecasting prices from level-i quotes in the presence of hidden
  liquidity.
\newblock {\em Algorithmic Finance}, 1(1):35--43.

\bibitem[\protect\astroncite{Biais et~al.}{1995}]{biais1995empirical}
Biais, B., Hillion, P., and Spatt, C. (1995).
\newblock An empirical analysis of the limit order book and the order flow in
  the paris bourse.
\newblock {\em the Journal of Finance}, 50(5):1655--1689.

\bibitem[\protect\astroncite{Biais and Weill}{2009}]{biais2009liquidity}
Biais, B. and Weill, P.-O. (2009).
\newblock Liquidity shocks and order book dynamics.
\newblock Technical report, National Bureau of Economic Research.

\bibitem[\protect\astroncite{Bouchaud et~al.}{2009}]{bouchaud2009markets}
Bouchaud, J.-P., Farmer, J.~D., and Lillo, F. (2009).
\newblock How markets slowly digest changes in supply and demand.
\newblock In {\em Handbook of financial markets: dynamics and evolution}, pages
  57--160. Elsevier.

\bibitem[\protect\astroncite{Cont and De~Larrard}{2013}]{cont2013price}
Cont, R. and De~Larrard, A. (2013).
\newblock Price dynamics in a markovian limit order market.
\newblock {\em SIAM Journal on Financial Mathematics}, 4(1):1--25.

\bibitem[\protect\astroncite{Cont and Mueller}{2019}]{cont2019stochastic}
Cont, R. and Mueller, M.~S. (2019).
\newblock A stochastic partial differential equation model for limit order book
  dynamics.
\newblock {\em Available at SSRN 3366536}.

\bibitem[\protect\astroncite{Cont et~al.}{2010}]{cont2010stochastic}
Cont, R., Stoikov, S., and Talreja, R. (2010).
\newblock A stochastic model for order book dynamics.
\newblock {\em Operations research}, 58(3):549--563.

\bibitem[\protect\astroncite{Doyne~Farmer et~al.}{2004}]{doyne2004really}
Doyne~Farmer, J., Gillemot, L., Lillo, F., Mike, S., and Sen, A. (2004).
\newblock What really causes large price changes?
\newblock {\em Quantitative finance}, 4(4):383--397.

\bibitem[\protect\astroncite{Golub et~al.}{2012}]{golub2012high}
Golub, A., Keane, J., and Poon, S.-H. (2012).
\newblock High frequency trading and mini flash crashes.
\newblock {\em Available at SSRN 2182097}.

\bibitem[\protect\astroncite{Jones}{2004}]{jones2004markov}
Jones, G.~L. (2004).
\newblock On the markov chain central limit theorem.
\newblock {\em Probability surveys}, 1(299-320):5--1.

\bibitem[\protect\astroncite{Logachov et~al.}{2018}]{logachov2018local}
Logachov, A., Logachova, O., and Yambartsev, A. (2018).
\newblock The local principle of large deviations for compound poisson process
  with catastrophes.
\newblock {\em arXiv preprint arXiv:1806.07459}.

\bibitem[\protect\astroncite{Logachov et~al.}{2019}]{logachov2019large}
Logachov, A., Logachova, O., and Yambartsev, A. (2019).
\newblock Large deviations in a population dynamics with catastrophes.
\newblock {\em Statistics \& Probability Letters}, 149:29--37.

\bibitem[\protect\astroncite{Menshikov and
  Petritis}{2014}]{menshikov2014explosion}
Menshikov, M. and Petritis, D. (2014).
\newblock Explosion, implosion, and moments of passage times for
  continuous-time markov chains: a semimartingale approach.
\newblock {\em Stochastic Processes and Their Applications}, 124(7):2388--2414.

\bibitem[\protect\astroncite{Meyn and Tweedie}{2012}]{meyn2012markov}
Meyn, S.~P. and Tweedie, R.~L. (2012).
\newblock {\em Markov chains and stochastic stability}.
\newblock Springer Science \& Business Media.

\bibitem[\protect\astroncite{Popov}{1977}]{popov1977geometric}
Popov, N. (1977).
\newblock Geometric ergodicity conditions for countable markov chains.
\newblock In {\em Doklady Akademii Nauk}, volume 234, pages 316--319. Russian
  Academy of Sciences.

\bibitem[\protect\astroncite{Roll}{1984}]{roll1984simple}
Roll, R. (1984).
\newblock A simple implicit measure of the effective bid-ask spread in an
  efficient market.
\newblock {\em The Journal of finance}, 39(4):1127--1139.

\bibitem[\protect\astroncite{Smith et~al.}{2003}]{smith2003statistical}
Smith, E., Farmer, J.~D., Gillemot, L.~s., Krishnamurthy, S., et~al. (2003).
\newblock Statistical theory of the continuous double auction.
\newblock {\em Quantitative finance}, 3(6):481--514.

\end{thebibliography}

\clearpage
%%--------------------------------------------------
%\appendix
%\section{Appendix}\label{macros}

%\vspace{4cm}

%\begin{figure}[h!]
		%\centering
	           %\includegraphics[scale=0.5]{images/cross_Q.pdf}
		% \caption{}
                      %\label{C_Q_bid_ask}
	%\end{figure}  

%-------------------------------------------------------

%\clearpage	

\end{document}